\documentclass[12pt]{article}

\usepackage[dvips]{graphicx}
\usepackage[english]{babel}
\usepackage{color}
\usepackage{amsmath}
\usepackage{rotate}
\usepackage{lscape}
\usepackage[latin1]{inputenc}
\usepackage{amssymb}
\usepackage{amsthm}
\newtheorem{theorem}{Theorem}[section]
\newtheorem{lemma}[theorem]{Lemma}
\newtheorem{proposition}[theorem]{Proposition}

\newtheorem{definition}[theorem]{Definition}

\title{A class of regression models for parallel and series systems with a random number of components}
\author{Alice L. Morais and Silvia L. P. Ferrari\\
\footnotesize{Departamento de Estatística, Universidade de São Paulo}}
\date{}

\begin{document}
\maketitle

\begin{abstract}

In this paper we extend the Weibull power series (WPS) class of distributions and named this new class as extended Weibull power series (EWPS) class of distributions. The EWPS distributions are related to series and parallel systems with a random number of components, whereas the WPS distributions (Morais and Barreto-Souza, 2011) are related to series systems only. Unlike the WPS distributions, for which the Weibull is a limiting special case, the Weibull law is a particular case of the EWPS distributions. We prove that the distributions in this class are identifiable under a simple assumption. We also prove stochastic and hazard rate order results and highlight that the shapes of the EWPS distributions are markedly more flexible than the shapes of the WPS distributions. We define a regression model for the EWPS response random variable to model a scale parameter and its quantiles. We present the maximum likelihood estimator and prove its consistency and normal asymptotic distribution. Although the 
construction of this class was motivated by series and parallel systems, the EWPS distributions are suitable for modeling a wide range of positive data sets. To illustrate potential uses of this model, we apply it to a real data set on the tensile strength of coconut fibers and present a simple device for diagnostic purposes.\\


\end{abstract}

\section{Introduction}

Reliability studies generally focus on the study of the failure of certain experimental units. It can often be assumed that there is a mechanism that leads to the failure of these units, e.g., a series or a parallel system. For instance, consider coconut fibers as the experimental units, and their rupture as their failure. The tension when the coconut fiber breaks can be interpreted as a failure of a parallel system because the microscopic architecture of this material is a collection  of smaller fibers. The rupture of this material happens after the rupture of all of the smaller fibers, which characterizes a parallel system. In this case, the number of components in the system is unknown. Some models have been proposed in the literature for modeling the time to failure of series and parallel systems, and many studies consider a fixed number of system components. Because real systems may be complex, it would often be more appropriate to consider an unknown amount of components.

Nakagawa and Zhao (2012) presented a model for the time to failure of parallel systems assuming a zero-truncated Poisson number of components. Marshall and Olkin (1997) defined a class of distributions from the minimum and the maximum of a geometric number of independent and identically distributed (iid) random variables. When these random variables are positive, the resulting distribution is related to series and parallel systems. Kus (2007) constructed a distribution connected to series systems based on the minimum of a Poisson number of iid exponential random variables. Crescenzo and Pellerey (2011) provided stochastic results for the time to failure of series and parallel systems with non-identically distributed components.
The limiting distribution of the maximum of a random number of independent random variables was discussed by Barndorff-Nielsen (1964). Barakat and El-Shandidy (2004) found the asymptotic behavior of general order statistics from a random-sized sample. 

Models based on series and parallel systems with a random number of components have been used not only in material strength studies, but also, in medical research and other fields. In associated medical research studies, the series and parallel systems are called the first and the last latent activation schemes, respectively. Yakovlev et al. (1993) proposed a framework for the first activation scheme. Cooner et al. (2007) illustrated some uses of both the first and the last latent activation schemes assuming the possibility of a cure rate.

Morais and Barreto-Souza (2011) introduced the Weibull power series (WPS) class of distributions, which are related to the time to failure of a series system with a random number of components. These researchers assumed that the failure times of the system components are independent and follow a Weibull distribution and that the unknown number of components follows a discrete power series distribution. Here we propose an extension of this class to include parallel systems to yield a more flexible class of distributions, which were denoted by the extended Weibull power series (EWPS) class of distributions. The WPS distributions have some restrictions that are relaxed after the proposed extension.

In this paper we first introduce the EWPS class of distributions and derive some of its properties. We then propose a EWPS regression model. Our approach is focused on modeling the scale parameter when the response is assumed to be an EWPS random variable. The scale parameter is of particular practical interest because it is directly proportional to the quantiles of the response variable.

The paper is organized as follows. In Section 2, we provide a brief review of WPS distributions. In Section 3, we define the EWPS class of distributions and derive some of its properties. In Section 4, we define a regression model with EWPS distributed response, discuss its estimation based on the maximum likelihood, and present asymptotic properties of the estimators. In Section 5, we present a strategy to infer the quantiles from a simple 
transformation of the scale parameter. In Section 6, we present a real data application of the EWPS regression model to illustrate potential 
uses of the new model and present a simple device for diagnostic purposes. 
In Section 7, we discuss larger classes of models that include EWPS distributions. In Section 8, we present some concluding remarks.

\section{Brief review of WPS distributions}

The WPS distributions are constructed from a composition between the discrete power series class of distributions and the Weibull law as follows. Let $a_n\geq 0$ for $n\in \mathbb{N}$ such that $a_1>0$ and 
\begin{eqnarray} \label{fpseries}
C(\theta)=\sum_{n=1}^\infty a_n \theta^n, \quad \forall \theta\in (-s,s),
\end{eqnarray}
where $s>0$ is the radius of convergence. Consider the function $p:\mathbb{N}\rightarrow \mathbb{R}$ given by
\begin{equation}\label{pseries}
p(n;\theta)=\frac{a_n\theta^n}{C(\theta)},\,\,\,\,\theta\in (-s,s),\quad n=1,2,\ldots.
\end{equation}
If $\theta>0$, $p(n;\theta)$ in (\ref{pseries}) is the probability function (pf) of a power series  distribution truncated at zero (Noack, 1950). We use the notation $N\sim\mbox{PS}(\theta;C)$ for the random variable $N$ with a pf $p(n;\theta)$ in (\ref{pseries}) with $\theta>0$. The Poisson, logarithmic, geometric, 
and binomial (where $m$ is the known number of replicates) distributions truncated at zero are special cases of the truncated power series distributions.

Let $Z_1, Z_2, \ldots$ be iid random variables with $Z_1\sim \mbox{Weibull}(\lambda,\alpha)$, i.e., $Z_1$ has a Weibull distribution with scale 
parameter $\lambda>0$, shape parameter $\alpha>0$, and probability density function (pdf)  
\begin{equation}\label{weibull}
g(z;\lambda,\alpha)=\alpha\lambda^{-\alpha} z^{\alpha-1}e^{-(z/\lambda)^\alpha}, \,\,\,\, z>0, \,\lambda>0, \,\alpha>0.
\end{equation}
Let $N\sim\mbox{PS}(\theta;C)$. Note that it is assumed that $\theta$ is positive in this case. The WPS class of distributions is defined by the marginal distribution of $Z_{(1)}=\min\{Z_1,\ldots,Z_N\}$ with the corresponding cumulative distribution function (cdf)
\begin{equation}\label{cdf}
F(y;\lambda,\alpha,\theta)=1-\frac{C\left(\theta S(y;\lambda,\alpha)\right)}{C(\theta)},\quad y>0,
\end{equation}
\noindent where $S(y;\lambda,\alpha)=\exp\{-(y/\lambda)^\alpha\}$ for $y>0$ is the survival function of the $\mbox{Weibull}(\lambda,\alpha)$ distribution. The WPS pdf is given by
\begin{equation}\label{dens_WPS}
f(y;\lambda,\alpha,\theta)=\frac{\theta g(y;\lambda,\alpha)}{C(\theta)} 
C'\left(\theta S(y;\lambda,\alpha)\right),\quad y>0,
\end{equation}
for $\lambda, \alpha>0$ and $0<\theta<s$. 

The WPS distributions arise, for example, in reliability studies. Assume that a
machine has $N$ unknown initial defects, and let $Z_i$ be the time to failure of the machine due to the $i$th defect, where $i=1,\ldots,N$.
If the $Z_i$'s are assumed to be iid variables with $Z_1\sim\mbox{Weibull}(\lambda,\alpha)$ and $N\sim \mbox{PS}(\theta;C)$, the time $Y$ to the first failure has a WPS pdf, as given in (\ref{dens_WPS}). In other words, $Y$ is the time to failure of a series system with a random number $N$ of components, where $N\sim\mbox{PS}(\theta;C)$.

In the next section we extend the parameter space for $\theta$ to include negative values and name the resulting class of distributions as extended Weibull power series (EWPS) distributions. After this extension is made, the original characterization of the WPS distributions, which is based on series systems, does not hold for $\theta<0$. We prove that there is a parallel system characterization for some EWPS distributions when $\theta<0$. This extension allows more flexibility in the shapes of the density and hazard functions. We also provide some results on the hazard rate order and the stochastic order to highlight the relevance of the proposed extension. Although the construction of this class is motivated by series and parallel systems, this model is suitable to data with positive support. 

The Weibull distribution is a limiting case of WPS distributions. In our proposed extension of this class, we define the Weibull law as the special case when $\theta=0$ and show that this definition is appropriate. For EWPS distributions, the Weibull law represents a system with a single component.

\section{The EWPS class of distributions}

Morais and Barreto-Souza (2011) introduced the WPS survival class of distributions, which are related to series systems with a random number of components. The WPS distributions have flexible density and hazard shapes but exhibit some restrictions that will be relaxed by the extension introduced in this section. For fixed scale and shape parameters ($\lambda$ and $\alpha$, respectively), we prove that the hazard function of any WPS distribution is always uniformly above the hazard function of the Weibull law. The extension of the WPS distribution proposed in this section includes distributions for which the opposite occurs.

As observed in Section 2, the WPS distributions are indexed by three parameters, namely $\lambda>0$, $\alpha>0$, and $\theta\in(0,s)$. The idea is to extend the WPS class of distributions to allow $\theta$ to assume negative values. To formalize the proposed extension, we provide the following proposition.

\begin{proposition}\label{prop1}
For each power series function $C(\cdot)$ in (\ref{fpseries}), let $S^*=\{\theta\in (-s,0): C'(\theta)=0\}$.
Then, for all $y>0$, $\lambda>0$, and $\alpha>0$, $f(y;\lambda,\alpha,\theta)$ in (\ref{dens_WPS}) is non-negative for all $\theta \in (s^*,0)$, where
\begin{equation}\label{def_s}
s^*=\left\{
\begin{array}[c]{cc}
	\max S^*,& {\rm if}\,\,\, S^*\neq\emptyset\\
	-s,& {\rm otherwise.}\\
\end{array}\right.
\end{equation}
\end{proposition}
\begin{proof} 
We first prove the existence of $s^*$. We have that $C'(\theta)=\sum_{i=1}^\infty na_n\theta^{n-1}\rightarrow a_1$ as $\theta\rightarrow 0$ and $a_1>0$. Hence, $\exists \varepsilon>0$ such that $C'(\theta)>0$ for $\theta \in(-\varepsilon,0)$. Therefore, because $C(\cdot)$ is a differentiable function, if $S^*\neq \emptyset$, the maximum of $S^*$ exists. This proves that $s^*$ is well defined. 

To complete the proof, it is sufficient to prove that $\theta C'(\theta b)/C(\theta)>0$, $\forall \theta\in(s^*,0)$, and $\forall b\in (0,1)$. Since $C(0)=0$ and from the construction of $s^*$, $C(\theta)$ is strictly negative or strictly positive for $\theta \in(s^*,0)$. If $C(\theta)<0$ for $\theta\in(s^*,0)$, $C(\theta)$ is strictly increasing in $\theta \in(s^*,0)$, which shows that $C'(\theta)>0$ for any $\theta \in(s^*,0)$. Hence, $\theta C'(\theta b)/C(\theta)>0$. If $C(\theta)>0$ for $\theta \in(s^*,0)$, the result follows analogously.
\end{proof}
Proposition \ref{prop1} states that it is possible to find an open interval $(s^*,0)$ such that  $f(y;\lambda,\alpha,\theta)$ in (\ref{dens_WPS}) is non-negative for all $\theta\in(s^*,0)$. Note that $\int_0^\infty f(y;\lambda,\alpha,\theta) dy=1 $ for any $\theta\in(s^*,0)$; hence,  $f(\cdot;\lambda,\alpha,\theta)$ is a density function. Therefore, it is possible to define an extension of the WPS distribution as follows.

\begin{definition}\label{def1}
For a given function $C(\cdot)$ in (\ref{fpseries}), the EWPS distribution with parameters $\lambda,\alpha>0$, $\theta \in (s^*,s)$, and $s^*$ as given in (\ref{def_s}) is defined by the pdf in (\ref{dens_WPS}) when $\theta \neq 0$, and by the pdf of the Weibull distribution given in (\ref{weibull}) when $\theta=0$.
\end{definition}

 We use the notation $Y\sim\mbox{EWPS}(\lambda,\alpha,\theta;C)$ when $Y$ is a random variable with the distribution given in Definition \ref{def1}. In Table \ref{tab1}, we provide useful quantities for the construction of some EWPS distributions. The last column is the name of the respective power series distribution with the pdf shown in (\ref{pseries}) when $\theta>0$. In cases 4 and 5, $m > 1$ is a known integer value, and for $\theta>0$, $m$ is the number of replicates of the binomial distribution or the fixed number of failures before $n$ successes of the negative binomial distribution. In case 6, $\mathcal{I}_{\{\mbox{odd}\}}(n)$ is the indicator function for odd $n$. If $C(\cdot)$ is chosen as in cases 1, 2, 3, 4, 5, and 6, the corresponding EWPS laws are called extended Weibull Poisson (EWP), extended Weibull geometric (EWG), extended Weibull logarithmic (EWL), extended Weibull binomial (EWB), extended Weibull negative binomial (EWNB), and extended Weibull logarithmic II (EWLII) distributions, 
respectively. The names of these special cases refer to the power series distribution used in the construction of the corresponding EWPS distributions for $\theta>0$. The extension of the parameter space for $\theta$ proposed here adds more flexibility to the density shapes, as observed in Figure 1.

\begin{table}[h!]
\begin{center} \small\caption{Useful quantities of some power series distributions.}\label{tab1}
\footnotesize
\begin{tabular}{ccccccc}
\hline\hline
&$a_n$&$C(\theta)$&$C'(\theta)$&$s^*$&$s$&Distribution (for $\theta>0$)\\
\hline
Case 1&$n!^{-1}$&$e^\theta-1$&$e^\theta$&$-\infty$&$\infty$&Poisson\\
Case 2&$n^{-1}$&$-\log(1-\theta)$&$(1-\theta)^{-1}$&$-1$&$1$&logarithmic\\
Case 3&$1$&$\theta(1-\theta)^{-1}$&$(1-\theta)^{-2}$&$-1$&$1$&geometric\\
Case 4&$\dbinom{m}{n}$&$(\theta+1)^m-1$&$m(\theta+1)^{m-1}$&$-1$&$\infty$&binomial\\
Case 5&$\dfrac{\Gamma(m+n-1)}{(n-1)!\Gamma(m)}$&$\theta(1-\theta)^{-m}$&$\dfrac{\{\theta(m-1)+1\}}{(1-\theta)^{m+1}}$&$\dfrac{1}{1-m}$&$1$&negative binomial\\
Case 6&${2}n^{-1}\mathcal{I}_{\{\mbox{odd}\}}(n)$&$\log\left(\dfrac{1+\theta}{1-\theta}\right)$&$\dfrac{2}{1-\theta^2}$&$-1$&$1$&logarithmic II\\
\hline\hline
\end{tabular}
\end{center}
\end{table}

\begin{figure}
	\centering
		\includegraphics[width=.33\textwidth]{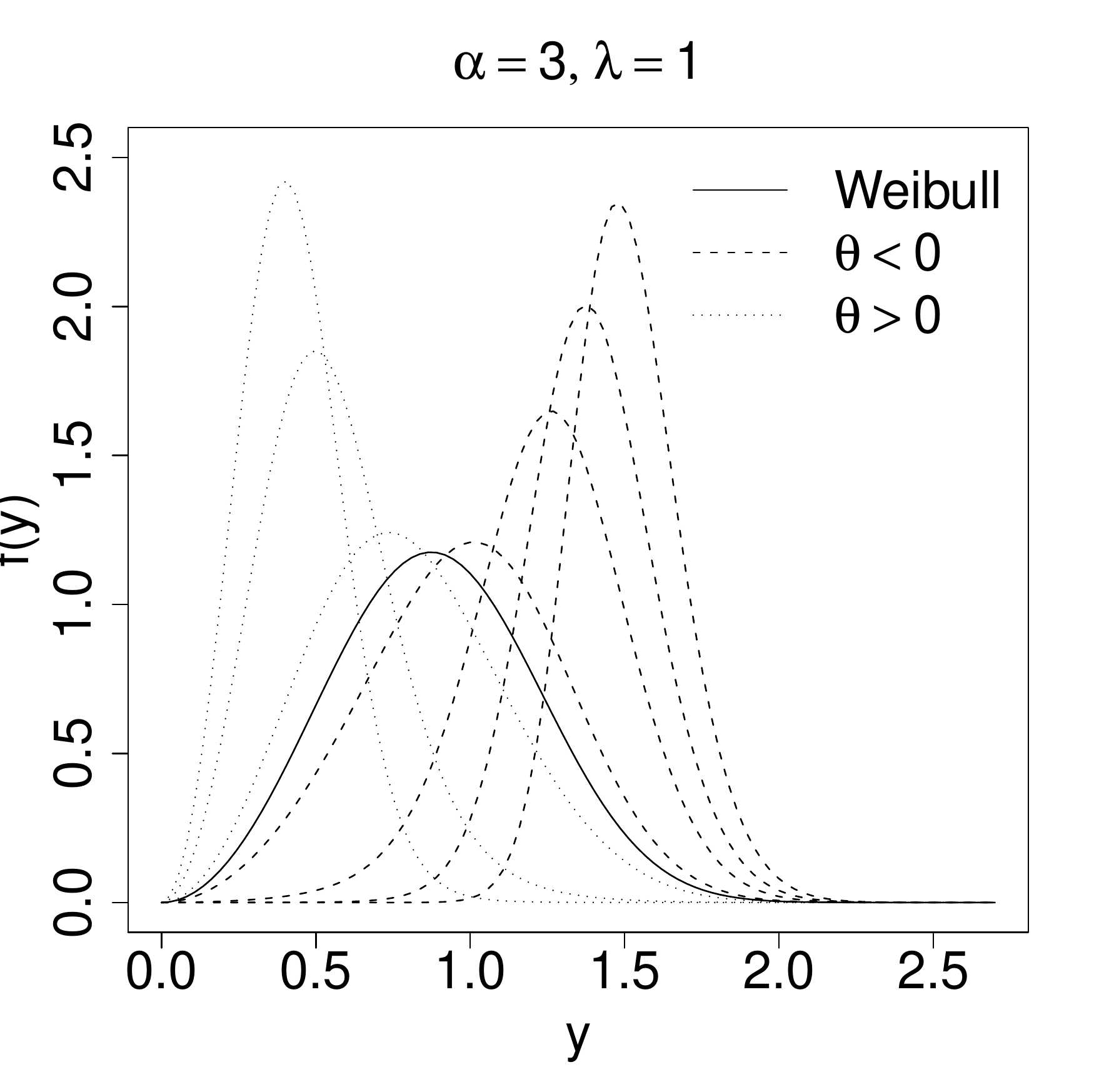}\includegraphics[width=.33\textwidth]{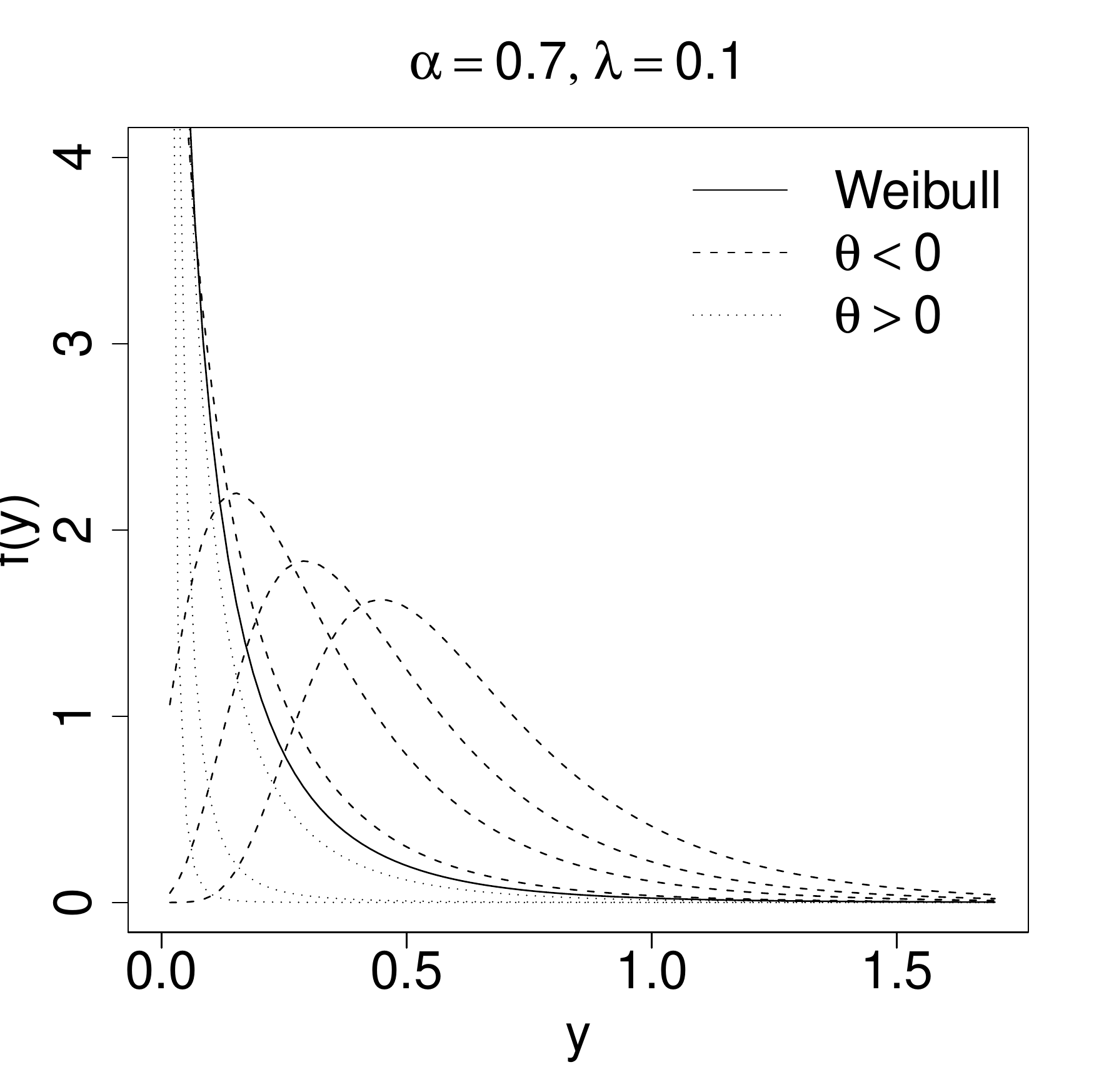}\includegraphics[width=.33\textwidth]{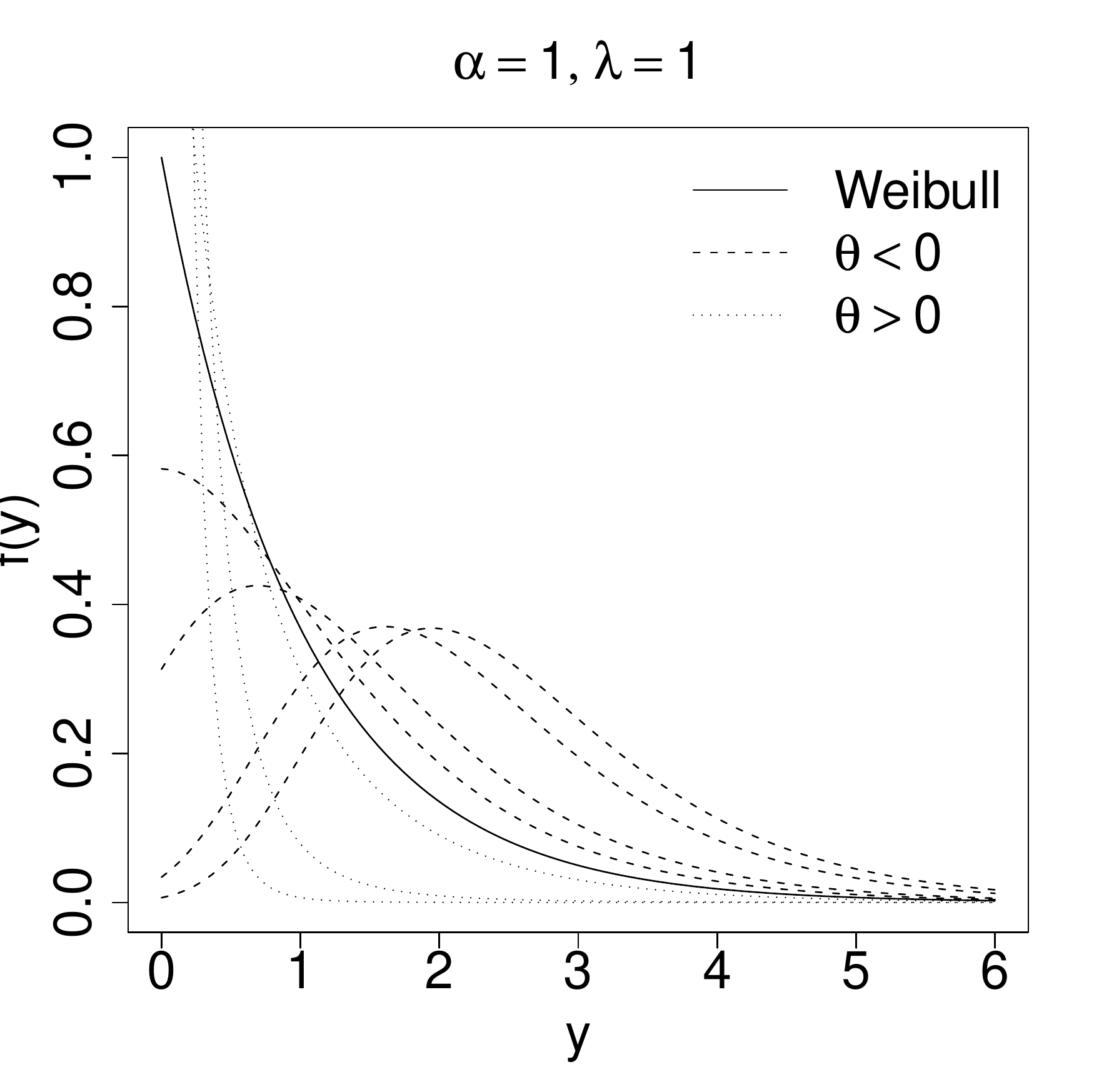}
	\label{fig:densWPcaso1}\caption{Density functions of the EWP distribution for some values of $\theta$, $\lambda$, and $\alpha$.}
\end{figure}

From  (\ref{dens_WPS}), we obtain that
\begin{equation}\label{expansao_dens}
f(y;\lambda,\alpha,\theta)=\sum_{n=1}^\infty \frac{a_n\theta^n}{C(\theta)}g\left(y;\lambda n^{-1/\alpha},\alpha\right),
\end{equation}
\noindent for $y>0$, $\lambda, \alpha >0$ and $\theta \in (s^*,0)\cup (0,s)$, where $g(y;\lambda,\alpha)$ is the pdf of  the Weibull distribution with parameters $\lambda$ and $\alpha$ given in (\ref{weibull}). Therefore, the EWPS densities are an infinite linear combination of Weibull densities. In particular, for $\theta>0$, the EWPS 
densities are infinite mixtures of Weibull densities with weights that are determined by power series laws. 
This property is helpful for obtaining 
the moments of the EWPS distributions and for proving identifiability. 
It follows from (\ref{expansao_dens}) and the Dominated Convergence Theorem that the $r$th moment 
of $Y\sim \mbox{EWPS}(\lambda,\alpha,\theta;C)$ for $\theta \neq 0$ is given by
\begin{equation*}
E(Y^r)=\frac{\Gamma(r/\alpha+1)\lambda^r}{C(\theta)}\sum_{n=1}^\infty \frac{a_n \theta^n}{n^{r/\alpha}}, \qquad r>0.
\end{equation*}
Note that all the moments are finite because $\sum_{n=1}^\infty \left |a_n \theta^n n^{-r/\alpha}\right|\leq C(\left|\theta\right|)<\infty$, for $\theta \in (s^*,s)$.

The domain of the density function in (\ref{dens_WPS}) regarded as a function of $\theta$ does not contain the point $\theta=0$. Recall that the EWPS distribution is  defined as the Weibull distribution when $\theta=0$. Since the Weibull law is a limiting special case of all EWPS distributions when $\theta\rightarrow 0$, as stated in the next proposition, the definition of the EWPS distribution as the Weibull distribution when $\theta=0$ is justified.

\begin{proposition}\label{prop3}
As $\theta\rightarrow 0$, the cdf given in (\ref{cdf}) converges to the cdf of the Weibull distribution with scale parameter $\lambda$ and shape parameter $\alpha$.
\end{proposition}
\begin{proof} Similar to the proof of Proposition \ref{prop3} shown by Morais and Barreto-Souza (2011).\end{proof}
Under a simple assumption for $C(\cdot)$, the EWPS distributions are identifiable, as stated in the next proposition.
\begin{proposition}
The EWPS distribution with parameters $\lambda, \alpha>0$, and $\theta \in (s^*,s)$ is identifiable if and only if $C(\cdot)$ is not an odd function.
\end{proposition}
\begin{proof} See Appendix.\end{proof}

When $C(\cdot)$ is an odd function, $f(y;\lambda,\alpha,\theta)=f(y;\lambda,\alpha,-\theta)$ for any $\theta\in(s^*,0)$. This is the only case in which the extension of the WPS distributions proposed in this paper may not be an advantage. In cases 1 through 5 in Table 1, $C(\cdot)$ is clearly not odd.

Some EWPS distributions exhibit a physical characterization when $\theta<0$. This characterization, however, is not based on a series system but on a parallel system.
Consider a parallel system with a random number $N$ of components and let $Z_1, Z_2\ldots$ be iid random variables, where $Z_i\sim \mbox{Weibull}(\lambda,\alpha)$ is the time to the failure of the $i$-th component. For $\theta<0$, under some conditions for the law of $N$ (see Proposition \ref{prop5}), the marginal distribution of $\max\{Z_1,\cdots,Z_N\}$ is an EWPS distribution. In other words, for $\theta>0$, the EWPS interpretation is based on the first component failure, whereas for some EWPS distributions when $\theta<0$, we have a characterization that considers the failure of the system when all of its components fail. 
In the next proposition, we give a sufficient condition for the existence of a parallel system characterization for the EWPS distributions when $\theta\in (s^*,0)$.

\begin{proposition}\label{prop5}
Let $Z_1, Z_2,\ldots$ be iid random variables with $Z_1\sim {\rm Weibull}(\lambda,\alpha)$. 
Let $N\sim{\rm PS}(t(\theta);C)$, where $t:(s^*,0)\rightarrow (0,s)$ and $s^*$ is given in (\ref{def_s}). Assume that $t(\cdot)$ satisfies 
\begin{equation}\label{cond1}
C^{(i)}(t(\theta))=\frac{(-1)^{i-1}a_{i}i!\theta^i C(t(\theta))}{t(\theta)^iC(\theta)},\,\,\mbox{for}\quad i\in \mathbb{N},
\end{equation} 
where $$C^{(i)}(\tau)=\frac{d^i C(\tau)}{d\tau^i}=\sum_{n=1}^\infty \frac{n!}{(n-i)!}a_n\tau^{n-i}, \quad \mbox{for}\quad \tau\in (-s,s).$$ 
Then, if $\theta\in (s^*,s)$, the marginal density function of $Z=\max\{Z_1,\ldots,Z_N\}$ is given in (\ref{dens_WPS}).
\end{proposition}
\begin{proof} See Appendix.\end{proof}

Proposition \ref{prop5} states that there is a parallel system characterization for some EWPS distributions when $\theta\in (s^*,0)$. For case 1 in Table \ref{tab1}, which refers to the Poisson distribution when $\theta>0$, the conditions of Proposition \ref{prop5} are satisfied by taking $t(\theta)=-\theta$, $\theta<0$. The same also occurs for cases 2 and 3 (which refer to the logarithmic and geometric distributions when $\theta>0$, respectively) by taking $t(\theta)=\theta/(\theta-1)$ for $\theta\in(-1,0)$. 

In Proposition \ref{prop4}, we provide some necessary conditions for the power series under which the corresponding EWPS distribution admits a parallel system characterization when $\theta<0$.

\begin{proposition}\label{prop4}
Let $t:(s^*,0)\rightarrow (0,s)$. If the density function in (\ref{dens_WPS}) for $\theta<0$ can be obtained through a parallel system characterization by taking $N\sim {\rm PS}(t(\theta);C)$, then the following statements hold:
\begin{description}
\item[(i)] If $a_n=0$, then $a_{m}=0$ for $m>n$;
\item[(ii)] $t(\cdot)$ is the unique solution of $E(N)={a_1\theta}/{C(\theta)};$
\item[(iii)] The function $t(\theta)$ is decreasing in $\theta$. 
\end{description}
\end{proposition}
\begin{proof} See Appendix.\end{proof}

The advantages of the parameter space extension of the WPS distributions go beyond the possibility of a parallel system characterization. For example, the hazard function of the EWPS distributions for $\theta<0$ may have shapes that the EWPS distributions for $\theta>0$ do not allow. The hazard function $r(\cdot;\lambda,\alpha,\theta)$ of the EWPS$(\lambda,\alpha,\theta;C)$ distribution is given by 

\begin{equation*}
r(y;\lambda,\alpha,\theta)=\left\{
\begin{array}[pos]{ll}
\dfrac{\alpha \theta y^{\alpha-1} }{\lambda^\alpha}\exp\left\{-\left(\dfrac{y}{\lambda}\right)^\alpha\right\}
\dfrac{C'\left(\theta \exp\left\{-\left(\dfrac{y}{\lambda}\right)^\alpha\right\}\right)}{C\left(\theta \exp\left\{-\left(\dfrac{y}{\lambda}\right)^\alpha\right\}\right)},& \ \theta\neq 0,\\
\dfrac{\alpha}{ \lambda^{\alpha}}y^{\alpha-1} ,& \ \theta=0,
\end{array}\right.
\end{equation*}

\noindent for $y>0,$ $\lambda>0,$ $\alpha>0$, and $\theta\in(s^*,s)$. Figure \ref{haz_fig} presents plots of the hazard function of the extended Weibull Poisson law for some choices of $\theta$, $\lambda$, and $\alpha$. 

The plots illustrate possible behaviors of the EWP hazard functions: J, inverted J, uniform, monotone, and non-monotone shapes. Note that the plotted curves that correspond to $\theta<0$ ($\theta>0$) are uniformly above (below) the Weibull hazard function. 

\begin{figure}[h]
	\centering
  \includegraphics[width=.330\textwidth]{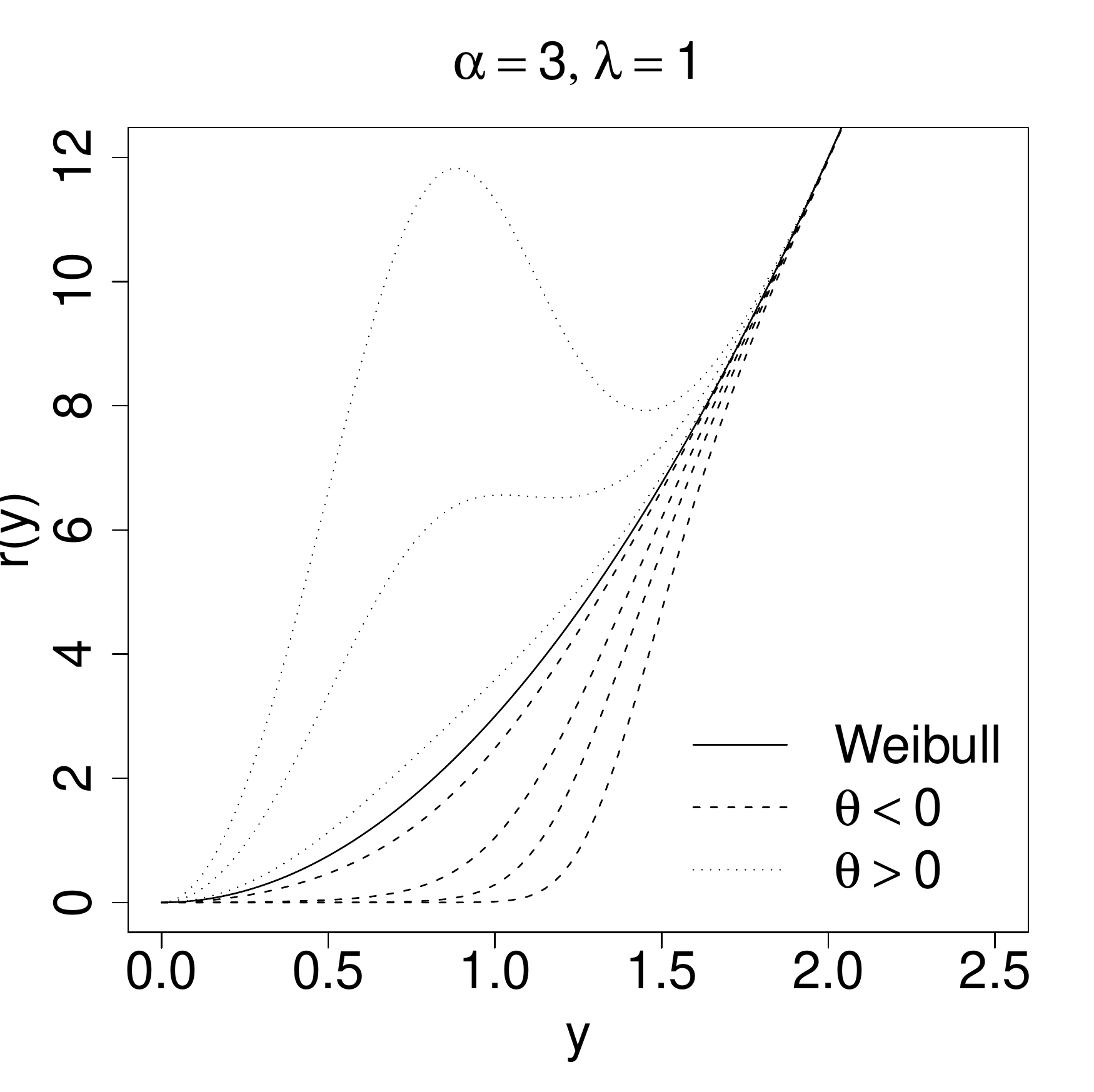}\includegraphics[width=.330\textwidth]{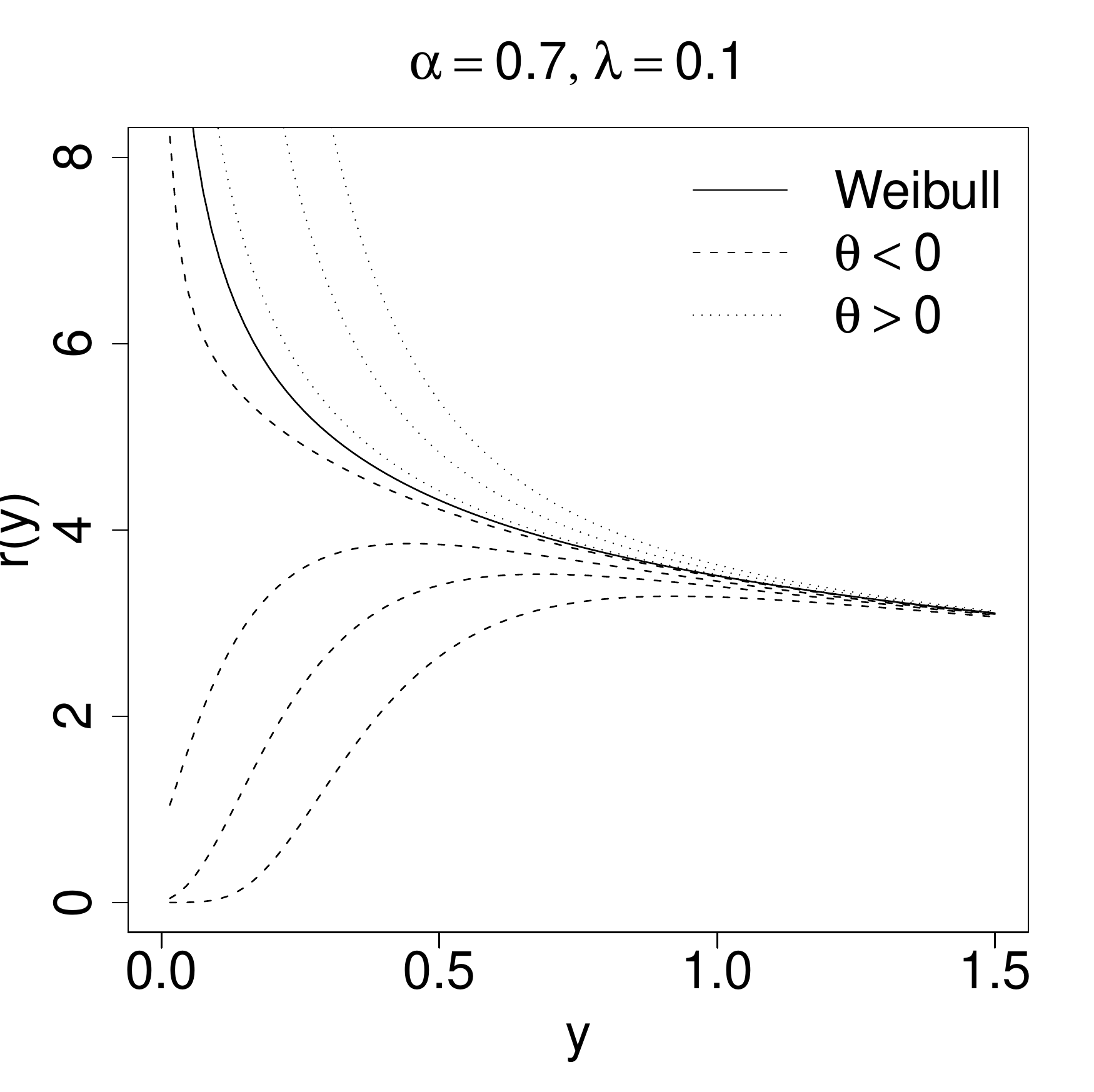}\includegraphics[width=.330\textwidth]{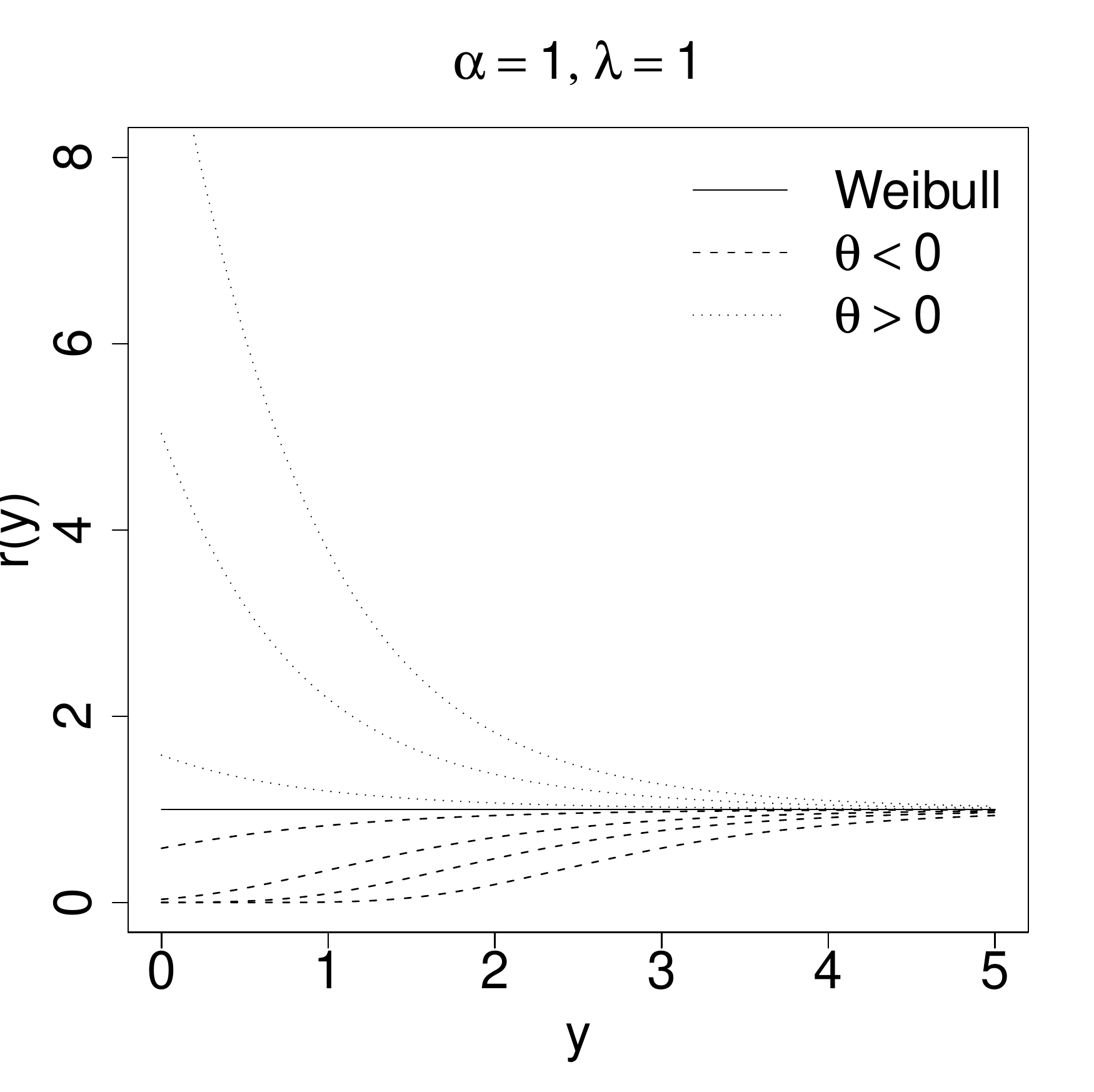}
	\label{haz_fig}\caption{Hazard functions of the EWP distribution for some values of $\theta$, $\lambda$, and $\alpha$.}
\end{figure}

For the EWP and some other EWPS distributions, there is an order of the hazard function regarded as a function of $\theta$ that explains the behavior of the curves presented in Figure 2. In Proposition 2.9, we give stochastic order and hazard rate order properties of EWPS distributions. 

\begin{definition}\label{def0}

\begin{description}
\item[(i)] A random variable $X_1$ is said to be smaller than $X_2$ in the stochastic order, and we write $X_1\leq_{_{st}}X_2$, if $P(X_1>x)\leq P(X_2>x)$, $\forall x\in\mathbb{R}$. 
\item[(ii)] Let $X_1$ and $X_2$ be non-negative random variables with absolutely continuous distribution functions. Let $r_1(\cdot)$ and $r_2(\cdot)$ be their respective hazard functions. If $r_1(x)\geq r_2(x)$ $\forall x\in\mathbb{R}$, $X_1$ is said to be smaller than $X_2$ in the hazard rate order, and we write $X_1\leq_{_{hr}}X_2$. 
\end{description}
\end{definition}

\begin{proposition}\label{prop7}
Let $Y_{\theta_1}\sim {\rm EWPS}(\lambda,\alpha,\theta_1;C)$ and $Y_{\theta_2}\sim {\rm EWPS}(\lambda,\alpha,\theta_2;C)$ with $s^*<\theta_1<\theta_2<s$. If
\begin{description}
\item [(i)] $0\leq\theta_1<\theta_2$ or
\item [(ii)] $\theta_1<0$ and ${\rm EWPS}(\lambda,\alpha,\theta;C)$ admits a parallel system characterization when $\theta<0$,
\end{description}
then $Y_{\theta_2}\leq_{st}Y_{\theta_1}$ and $Y_{\theta_2}\leq_{hr}Y_{\theta_1}$.
\end{proposition}
\begin{proof} See Appendix.\end{proof}

For $\theta\geq 0$ and fixed values for $\lambda$ and $\alpha$, there is always a stochastic order and a hazard rate order. In this case, which corresponds to the WPS distributions, the hazard function is uniformly below the hazard function of the Weibull law. For $\theta<0$, depending on the choice of (\ref{fpseries}), these orders still hold, and the hazard rate function of the EWPS is uniformly above the hazard function of the Weibull distribution. This result emphasizes the relevance of extending the parameter space of $\theta$. 

Note that the hazard functions plotted in Figure 2 are close to the hazard function of the Weibull distribution for large values of $y$. This is a valid property for all EWPS distributions as stated in the following proposition.

\begin{proposition} Let $r_0(\cdot;\lambda,\alpha)$ be the hazard function of the ${\rm Weibull} (\lambda,\alpha)$ distribution, and let $r(\cdot;\lambda,\alpha,\theta)$ be the hazard function of the ${\rm EWPS}(\lambda,\alpha,\theta;C)$ distribution. Then, $|r(y;\lambda,\alpha,\theta)-r_0(y;\lambda,\alpha)|\rightarrow 0$ as $y\rightarrow \infty$ for all $\lambda,\alpha>0$ and $\theta\in(s^*,s)$.
\end{proposition}
\begin{proof} See Appendix.\end{proof}

\section{EWPS regression model and estimation}
 
Let $Y_1,\ldots,Y_n$ be a vector of independent random variables. Let $X=(x_1,\ldots,x_n)^\top$ be a fixed $n\times k$ matrix of covariates with 
$x_i=(x_{i1},\ldots,x_{ik})^\top$ and $\eta_i=x_i^\top\beta$ for $i=1,\ldots,n$, where $\beta \in \mathbb{R}^k$ is a vector of unknown parameters. 
The extended Weibull power series regression model is defined by
$Y_i\stackrel{ind}{\sim} {\rm EWPS}(\lambda_{_{i}},\alpha,\theta;C)$ for $i=1,\ldots,n$, with 
\begin{equation}\label{link}
h(\lambda_{_{i}})=\eta_i,
\end{equation}
where $h(\cdot)$ is an invertible and three times differentiable link function that maps $(0,\infty)$ in $\mathbb{R}$, and $\alpha>0$, $\theta\in(s^*,s)$ and $\beta$ are unknown parameters. A possible choice for $h(\cdot)$ is the logarithmic function.

We now discuss estimation by maximum likelihood in the EWPS regression model. 
Let $\mathcal{Y}=(Y_1,\ldots,Y_n)$ be a vector of $n$ independent random variables following the EWPS regression model. Let $\Theta=(\beta^\top,$ $\alpha,\theta)^\top$ be the parameter vector.
The total log-likelihood function is given by
\begin{eqnarray*}
\ell\equiv\ell(\Theta;\mathcal{Y},X)=n\log\alpha+\sum_{i=1}^n\log W_{i}- \sum_{i=1}^n W_i+\sum_{i=1}^n L_{1i}(\theta)+c,
\end{eqnarray*}
where $W_i=(Y_i/\lambda_i)^\alpha$, $c$ is a constant that does not depend on the parameters, and 
\begin{equation*}
L_{1i}(\theta)=\left\{
\begin{array}[2]{ll}
\log\left [\dfrac{C'\left(\theta e^{-W_i}\right)}{C(\theta)}\theta\right ],&\theta\neq 0,\\
0, & \theta=0.
\end{array}
\right.
\end{equation*}
The associated score vector is given by $U_n(\Theta)=(\partial\ell/\partial\beta^\top,\partial\ell/\partial\alpha,\partial\ell/\partial\theta)^\top$, 
where
\begin{eqnarray*}
\frac{\partial \ell}{\partial \beta} &=& \alpha X^\top D_1Y^*,\\
\frac{\partial \ell}{\partial \alpha}&=& \dfrac{n}{\alpha}+\dfrac{1}{\alpha}{\bf 1}^\top D_2Y^*,\\
\frac{\partial \ell}{\partial \theta}&=& \left\{
\begin{array}[2]{ll}
\dfrac{n}{\theta}-n\left(\dfrac{C'(\theta)}{C(\theta)}\right)+\dfrac{1}{\theta}{\rm tr}(D_3),&\theta\neq 0,\\
\dfrac{a_2}{a_1}{\rm tr}(2D_4-I),& \theta=0,
\end{array}
\right.
\end{eqnarray*}
In these equations, ${\bf 1}\equiv{\bf 1}_{n\times 1}$ is an $n\times 1$ column vector of ones, $I$ is the $n\times n$ identity matrix, and $D_1$, $D_2$, $D_3$ and $D_4$ are $n\times n$ diagonal matrices given by
\begin{eqnarray*}
D_1&=&\mbox{diag}\left(\frac{1}{\lambda_1}\frac{d \lambda_{1}}{d\eta_1},\ldots,\frac{1}{\lambda_n}\frac{d \lambda_{n}}{d\eta_n}\right),\\
D_2&=&\mbox{diag}\left(\log W_1,\ldots,\log W_n\right),\\
D_3&=&\mbox{diag}\left(L_{21}(\theta),\ldots,L_{2n}(\theta)\right),\\
D_4&=&\mbox{diag}\left(e^{-W_1},\ldots,e^{-W_n}\right),
\end{eqnarray*} 
where $W_i=({Y_i}/{\lambda_{i}})^\alpha$, $Y^*$ is a $n\times 1$ column vector given by $$Y^*=\left(W_1\left(1+L_{21}(\theta)\right)-1,\ldots,W_n\left(1+L_{2n}(\theta)\right)-1\right )^\top,$$ and
\begin{equation*}
L_{2i}(\theta)=\left\{
\begin{array}[2]{ll}
\dfrac{C''(\theta e^{-W_i})}{C'(\theta e^{-W_i})}\theta e^{-W_i}, & \theta\neq 0,\\
0, & \theta=0.
\end{array}\right .
\end{equation*}
The maximum likelihood estimate (MLE) of $\Theta$, which is denoted by $\widehat{\Theta}$, is obtained by solving the nonlinear system of equations $U_n(\Theta)=0$. 
For some EWPS models and some choices of the link function, the solution of this system can be simplified. For instance, if the chosen link function is the logarithmic function, $D_1$ is the $n\times n$ identity matrix.
Let
\begin{equation}\label{hessian}
K_n(\Theta)=-\frac{\partial ^2\ell}{\partial\Theta\partial\Theta^\top}
\end{equation}
be the total observed information matrix. Its elements are given by
\begin{eqnarray*}
\frac{\partial^2\ell}{\partial \beta \partial \beta^\top}&=&\alpha X^\top \left\{D_7D_8-D_1^2\left[D_8+\alpha(D_8+I)-\alpha D_5^2( D_6+ D_3)\right]\right\} X,\\
\frac{\partial^2\ell}{\partial \alpha \partial \beta^\top}&=&X^\top D_1\left\{D_8+ D_2 D_5\left[D_3+D_5\left(I-D_6\right)-D_8\right]\right\}{\bf 1},\\ 
\frac{\partial^2\ell}{\partial \alpha^2}&=&-\frac{1}{\alpha^2}{\rm tr}\left \{ D_2^2 \left[D_8+I-D_5^2( D_6+ D_3)\right]+I \right \},\\
\frac{\partial^2\ell}{\partial \theta \partial \beta^\top}&=&\left\{
\begin{array}[2]{ll}
\alpha \theta^{-1}X^\top D_1 D_5(D_6+D_3){\bf 1},&\theta\neq 0,\\ 
\dfrac{\alpha 2 a_{2} }{a_{1} }X^\top D_1 D_4 D_5 {\bf 1},&\theta=0,\\
\end{array}\right .\\
\frac{\partial^2\ell}{\partial \alpha \partial \theta}&=&
\left\{
\begin{array}[2]{ll}
-\theta^{-1}\alpha^{-1}{\rm tr}( D_2D_5(D_6+D_3)),&\theta\neq 0,\\ 
\dfrac{2 a_{2}}{a_{1} }{\rm tr} ( D_2D_4D_5),&\theta= 0,
\end{array}\right .\\
\frac{\partial^2\ell}{\partial \theta^2}&=&\left\{
\begin{array}[2]{ll}
-n\left[\dfrac{1}{\theta^2}+\dfrac{C''(\theta)C(\theta)-C'(\theta)^2}{C(\theta)^2}\right]+\dfrac{1}{\theta^{2}}{\rm tr}( D_6),&\theta\neq 0,\\ 
\dfrac{a_{2}(a_{2}-2a_{1})}{a_{1}^2}-\dfrac{4a_{2}^2}{a_{1}^2}{\rm tr}( D_4^2),&\theta=0,
\end{array}\right .
\end{eqnarray*}
where $D_5$, $D_6$, and $D_7$ are $n\times n$ diagonal matrices given by
\begin{eqnarray*}
D_5&=&\mbox{diag}\left(W_1,\ldots,W_n\right),\\
D_6&=&\mbox{diag}\left (L_{31}(\theta),\ldots,L_{3n}(\theta)\right ),\\
D_7&=&\mbox{diag}\left(\frac{1}{\lambda_1}\frac{d^2 \lambda_{1}}{d\eta_1^2},\ldots,\frac{1}{\lambda_n}\frac{d^2 \lambda_{n}}{d\eta_n^2}\right),\\
D_8&=&\mbox{diag}\left(W_1\left(1+L_{21}(\theta)\right)-1,\ldots,W_n\left(1+L_{2n}(\theta)\right)-1\right ),
\end{eqnarray*} 
and 
\begin{equation*}
L_{3i}(\theta)=\left\{
\begin{array}[2]{ll}
\left[\dfrac{ C'''(\theta e^{-W_i})}{C'(\theta e^{-W_i})}-\left(\dfrac{C''(\theta e^{-W_i})}{C'(\theta e^{-W_i})}\right)^2\right]\theta^2 e^{-2W_i},& \theta\neq 0\\
0, &\theta=0.
\end{array}\right .
\end{equation*}
For the EWP distribution, $D_6$ is
the $n\times n$ matrix of zeros. If the chosen link function is the logarithmic function, $D_7$ is the $n\times n$ identity matrix.

We now state the consistency and asymptotic normality of the ML estimator.

\begin{lemma}\label{prop_cond}
Let $\Omega\subset \mathbb{R}^{k+2}$ be the parameter space, and let $\omega\subset \Omega$ be an open region that contains the true parameter vector $\Theta^{(0)}=(\beta^\top_0,\alpha_0,\theta_0)$. Then, the following statements hold.
\begin{enumerate}
\item[(i)] Each third-order derivative of the log-likelihood function exists and is dominated by an integrable function that does not depend on the parameters for all $\Theta\in\omega$.
\item[(ii)] $$E(Y^*)={\bf 0}_n, \quad E({\bf 1}^\top D_1 Y^*)=\frac{n}{\alpha} \quad \mbox{and}\quad E\left(\frac{C''\left(\theta e^{-W}\right) e^{-W}}{C'\left(\theta e^{-W}\right)}\right)=\frac{C'(\theta)}{C(\theta)}-\frac{1}{\theta},$$ where ${\bf 0}_n$ is an $n\times 1$ column vector of zeros. The expected value of each element of the score vector is then zero.
\item[(iii)] The integrals $E\left(U_n(\Theta)\, U_n(\Theta)\right)$ can be differentiated under the integral sign for all $\Theta\in\omega$; hence, $E(U_n(\Theta)U_n(\Theta)^\top)=E(K_n(\Theta))$.
\item[(iv)] \begin{eqnarray*}
\frac{1}{n}U_n(\Theta^{(0)})\rightarrow  {\bf 0}_{k+2} \quad \mbox{in probability, and} \quad
\frac{1}{n}K_n(\Theta^{(0)})\rightarrow  J \quad \mbox{in probability},
\end{eqnarray*}
\noindent where $J$ is a $(k+2)\times(k+2)$ finite matrix.
\end{enumerate}
\end{lemma}
\noindent {\it Proof: See Appendix.}

\begin{proposition} \label{prop_asymptotic}

Let $Y_i\stackrel{ind}{\sim} {\rm EWPS}(\lambda_i,\alpha,\theta;C) $, where $\lambda_i$ satisfies (\ref{link}). Let $\Omega\subset \mathbb{R}^{k+2}$ be the parameter space, and let $\omega\subset \Omega$ be an open region that contains the true parameter vector $\Theta^{(0)}=(\beta^\top_0,\alpha_0,\theta_0)$. Assume that the following conditions hold:
\begin{enumerate}
\item[C1 - ] $C(\cdot)$ is not an odd function. 
\item[C2 - ] The rank of $X$ is $k$.
\item[C3 - ] $\exists$ $m<\infty$ such that $|x_{ij}|<m$, for all $i=1,\ldots,n$ and $j=1,\ldots,k$.
\item[C4 - ] $J$ is positive definite for all $\Theta\in \omega$.
\end{enumerate}
Then, with probability tending to $1$ as $n\rightarrow \infty$, there exist solutions $\widehat\Theta_n$ of the likelihood equations such that 

\begin{enumerate}
\item $\widehat\Theta_n$ is consistent;
\item $\sqrt n (\widehat\Theta_n-\Theta)^\top \stackrel{ D}{\longrightarrow} N_{k+2}(0,J^{-1}),$
\noindent where $N_k(0,\Sigma)$ is a $k$-variate normal distribution with mean zero and covariance matrix $\Sigma$.

\end{enumerate}
\end{proposition} 
\begin{proof}
Under the true parameter vector, the score vector and the total observed information matrix depend on the response variable only through the iid random variables $W_1,\ldots,W_n$. From Lemma \ref{prop_cond}, the proof follows similarly to the proof of Theorem 5.1 described by Lehmann and Casella (1998, p. 463). 

\end{proof}

Let $(\Theta_1,\Theta_2)$ be a partition of the parameter vector and $d\leq k+2$ be the dimension of $\Theta_1$. Consider the null hypothesis $H_0:\Theta_1=\Theta^{(0)}$. From Lemma 3.1 and Proposition 3.2, we have that, under $H_0$, the asymptotic distribution of the score, likelihood ratio, and Wald statistics are $\chi_d^2$ (see Sen et al., 2011, p. 261). Therefore, these test statistics can be used to test the suitability of the EWPS regression against its main nested model, which is the Weibull regression model.

\section{Quantile estimation}

The quantile of order $\xi$ ($0<\xi<1$), $\xi$-quantile for short, of the distribution of a random variable $Y$ with cdf $F_{Y;\theta}(\cdot)$, 
which is denoted by $q_{_\xi}$, is the solution of
$
q_{_\xi}=\inf\{y:F_{Y;\theta}(y)\leq \xi\}.
$
If $Y$ has a continuous distribution, the $\xi$-quantile of $Y$ can be expressed as 
\begin{equation}\label{quantil}
q_{_\xi}=F_{Y;\theta}^{-1}(\xi).
\end{equation}

Let $q_{_\xi}$ be the $\xi$-quantile of a EWPS distribution for a fixed $\xi\in(0,1)$. It follows from (\ref{quantil}) and (\ref{cdf}) that the $\xi$-quantile can be written as
\begin{equation*}
q_{_\xi}=\lambda B_{_\xi}(\theta)^{1/\alpha},
\end{equation*}
where $B_{_\xi}(\theta)=-\log (C^{-1}((1-\xi)C(\theta))/\theta)$, for $\theta\neq 0$, and $C^{-1}(\cdot)$ is the inverse function of $C(\cdot)$. Note that $C(\cdot)$ is monotone for $\theta\in(s^*,s)$, which means that its inverse function is well defined.
For $\theta=0$, $B_{_\xi}(\theta)=\left(-\log(1-\xi)\right)^{1/\alpha}$.
Clearly, for any fixed $\xi$, $q_{_\xi}$ is a scale parameter. Quantiles of different orders obey the following proportionality 
relationship:
\begin{equation*}
\frac{q_{_\xi}}{q_{_{\xi'}}}=\left(\frac{B_\xi(\theta)}{B_{\xi'}(\theta)}\right)^{1/\alpha},\,\,\, \xi\neq\xi'.
\end{equation*}

Let $\widehat \beta$, $\widehat\theta$, and $\widehat\alpha$ be the MLEs of  $\beta$, $\theta$, and $\alpha$, respectively. From the invariance property of MLEs, we find that
$$\widehat q_{_{\xi,i}}=\widehat \lambda_i B_{_\xi}\big(\widehat\theta\big)^{1/\widehat\alpha}, $$
is the MLE of $q_{_{\xi,i}}$, where $\widehat \lambda_i=h^{-1}(\widehat\eta_i)$, and $q_{_{\xi,i}}$ is $\xi$-quantile of $Y_i$. Here, $\widehat \eta_i=x_i^\top \widehat\beta$. For the log-link function, i.e., $h(\lambda_i)=\log(\lambda_i)$, we can write $q_{_{\xi,i}}=\exp\{\beta^*_0+\beta_1 x_{1i}+\ldots+\beta_k x_{ki}\}$, where $\beta^*_0=\beta_0+\alpha^{-1}\log B_{_\xi}(\theta)$ is a modified intercept. 
In this case, the effect of the covariates on any quantile or scale parameter is the same.

The second-order Taylor series expansion of $\widehat q_{_{\xi,i}}$ around $(\widehat\beta^\top,\widehat\alpha,\widehat\theta)^\top=
(\beta^\top,\alpha,\theta)^\top$ is given by
\begin{eqnarray*}
\widehat q_{_{\xi,i}}&\approx& q_{_{\xi,i}} + \frac{\widehat\theta - \theta}{\alpha} \lambda_iB'_{_\xi}(\theta)B_{_\xi}(\theta)^{1/\alpha-1}-\frac{\widehat\alpha-\alpha}{\alpha^2}h^{-1}(\eta_i)B_{_\xi}(\theta)^{1/\alpha}\log\left(B_{_\xi}(\theta)\right)\\
&&+\frac{\partial \lambda_i}{\partial \eta_i}B_{_\xi}(\theta)^{1/\alpha}\sum_{j=1}^k (\widehat\beta_j-\beta_j)x_{ij}.\\
\end{eqnarray*}
An approximation for the variance of $\widehat q_{_{\xi,i}}$ is given by

\begin{equation}\label{var_q}
{\rm Var}(\widehat q_{_{\xi,i}})\approx E_i^\top \Sigma E_i ,
\end{equation}
where $\Sigma$ is the asymptotic covariance matrix of $\widehat\Theta=(\widehat\beta,\widehat\alpha,\widehat\theta)$, $E_i=(\varepsilon_{1i}^\top,\varepsilon_{2i},\varepsilon_{3i})^\top$ is a vector 
with elements given by $$\varepsilon_{1i}=\frac{d \lambda_i}{d \eta_i} B_{_\xi}(\theta)^{1/\alpha}x_i,\quad \varepsilon_{2i}=-\alpha^{-2}\lambda_i B_{_\xi}(\theta)^{1/\alpha}\log\left(B_{_\xi}(\theta)\right), \quad \varepsilon_{3i}=\lambda_i\alpha^{-1}B'_{_\xi}(\theta)B_{_\xi}(\theta)^{1/\alpha-1}.$$ Because $x_i=(x_{i1},\ldots,x_{ik})^\top$, $\varepsilon_{1i}$ is a vector of length $k$. In practice, $\Sigma$ can be estimated by $K_n^{-1}(\widehat\Theta)$ given in (\ref{hessian}).

\section{Application to real data}

In this section we illustrate an application of the EWPS regression model to a real data set from Tomczak (2010) on the tensile strengths of coconut fibers of different diameters and different lengths.
We fit the extended Weibull Poisson (EWP), Weibull logarithmic (EWL), and 
Weibull geometric (EWG) models, which are EWPS distributions with the function $C(\cdot)$ as defined in Table \ref{tab1}. The results were contrasted with the Weibull regression model. For this application, the corresponding regression models are specified using the log-link function, i.e., 
\begin{equation*}
\lambda_{_{i}}=\exp\{\beta_0+\beta_1x_{1i}+\beta_2 x_{2i}\}, \quad i=1,\ldots,225.
\end{equation*}
Accordingly, the regression model for the $\xi$-quantile is
\begin{equation*}
q_{_{\xi,i}}=\exp\{\beta_0^*+\beta_1x_{1i}+\beta_1 x_{2i}\},\quad i=1,\ldots,225,
\end{equation*}
where $\beta_0^*=\beta_0+\alpha^{-1}\log B_{_\xi}(\theta)$, $x_{1i}$ is the length in millimeters (mm) of the $i$th fiber, and $x_{2i}$ is the logarithm of the diameter (mm) of the $i$th fiber.
Because the logarithmic link function is employed, the effect of the covariate on the scale parameter and on all the quantiles is the same, as discussed in Section 4. 

To obtain the MLEs, we construct a profile log-likelihood function by fixing a grid of values of $\theta$. We start with $\theta=0$, i.e., the Weibull regression model. The estimated parameters for fixed $\theta=0$ are used as starting values to obtain the estimates for fixed $\theta=0.01$; the estimated parameters for fixed $\theta=0.01$ are used as starting values to obtain the estimates for fixed $\theta=0.02$, and so on, and similarly to negative values of $\theta$. Close to the global maximum, we evaluate the profile log-likelihood function at a finer grid of values of $\theta$. We use the same procedure for the negative values of $\theta$.
For each fixed value of $\theta$, we find the estimates of the parameters using the BFGS method implemented in the software R. Figure 3 presents plots of the profile log-likelihood for the EWL, EWG, and EWP regression models.

\begin{figure}[h]
	\centering
\includegraphics[width=.34\textwidth]{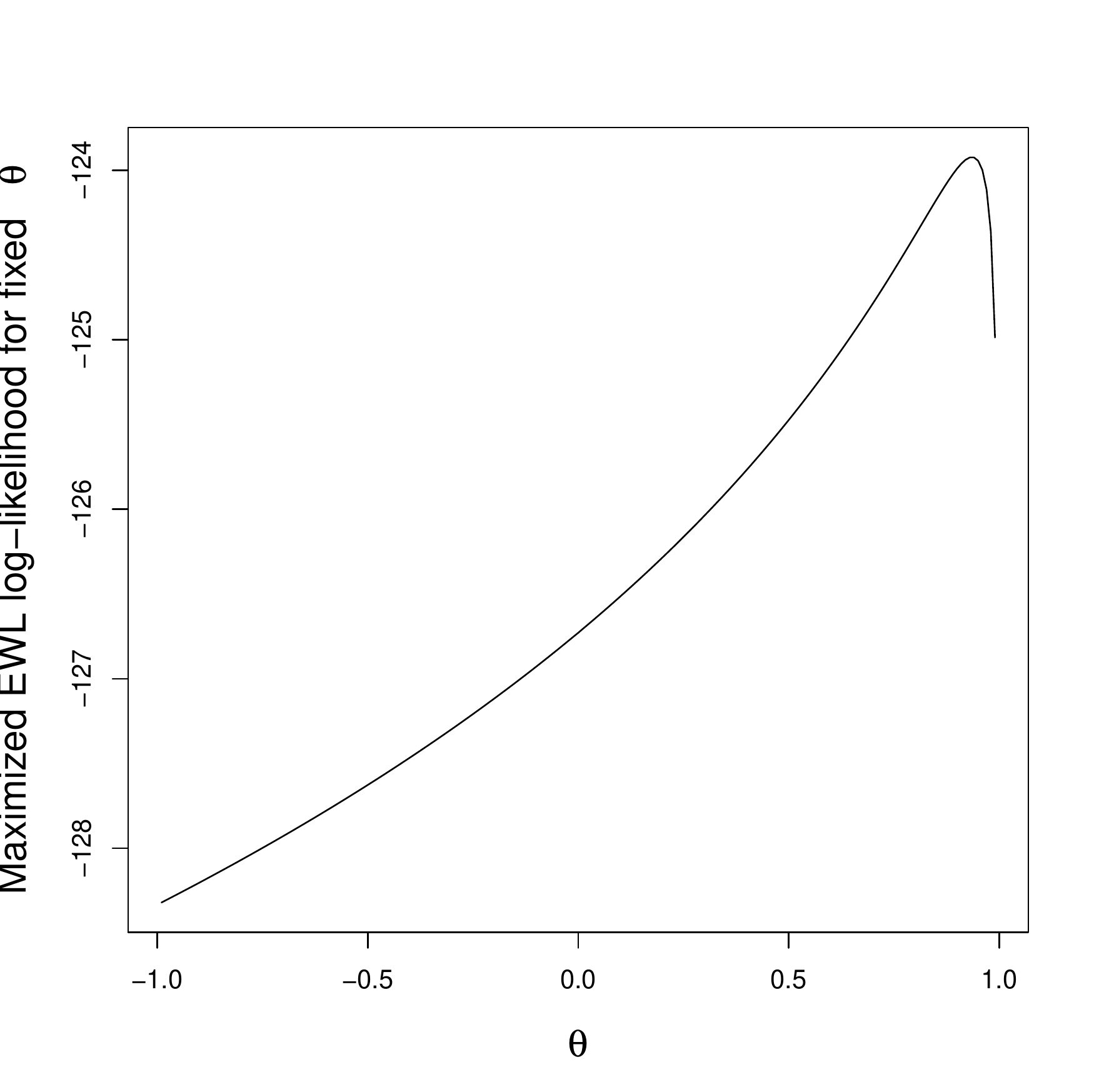}\includegraphics[width=.34\textwidth]{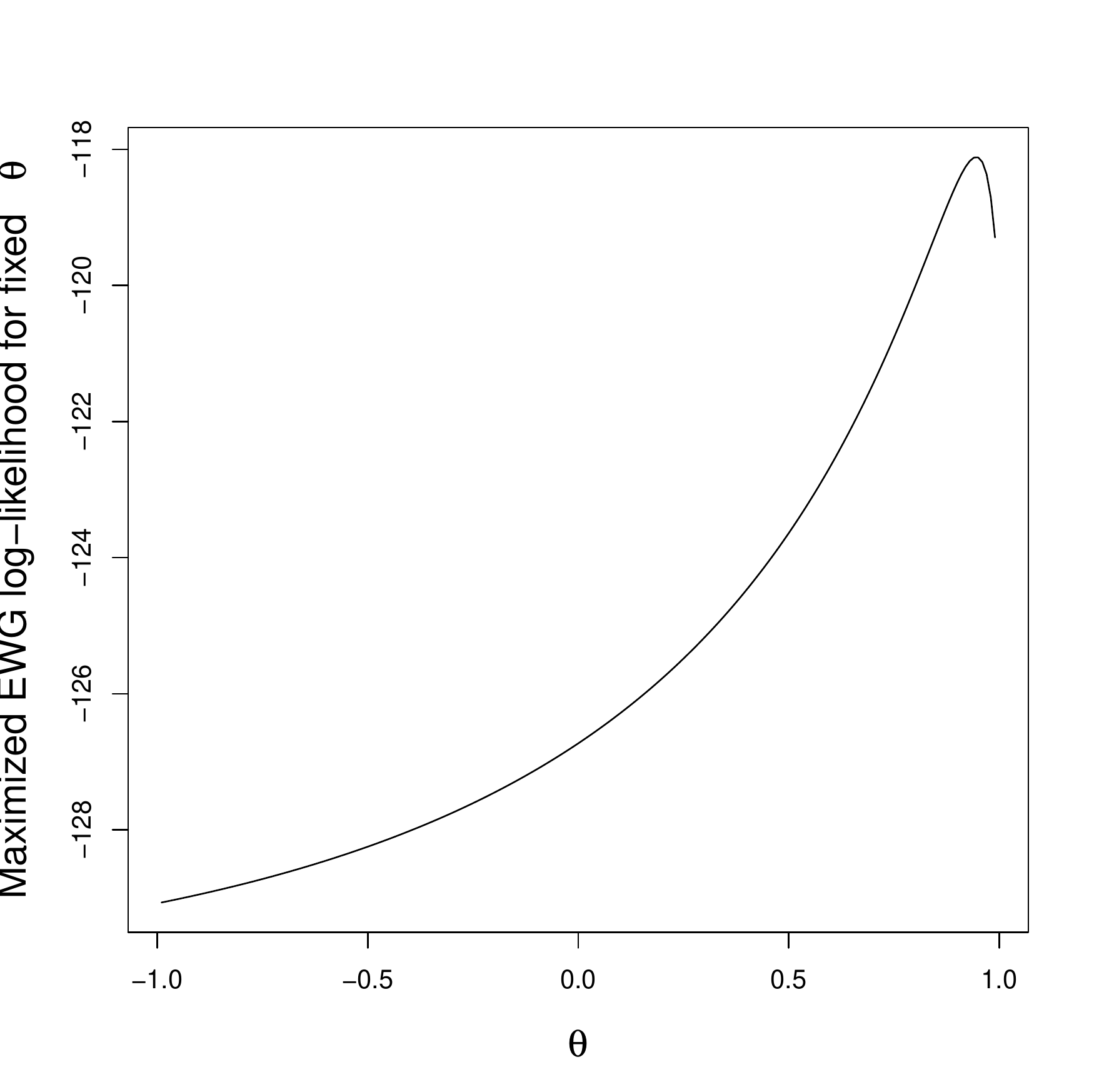}\includegraphics[width=.34\textwidth]{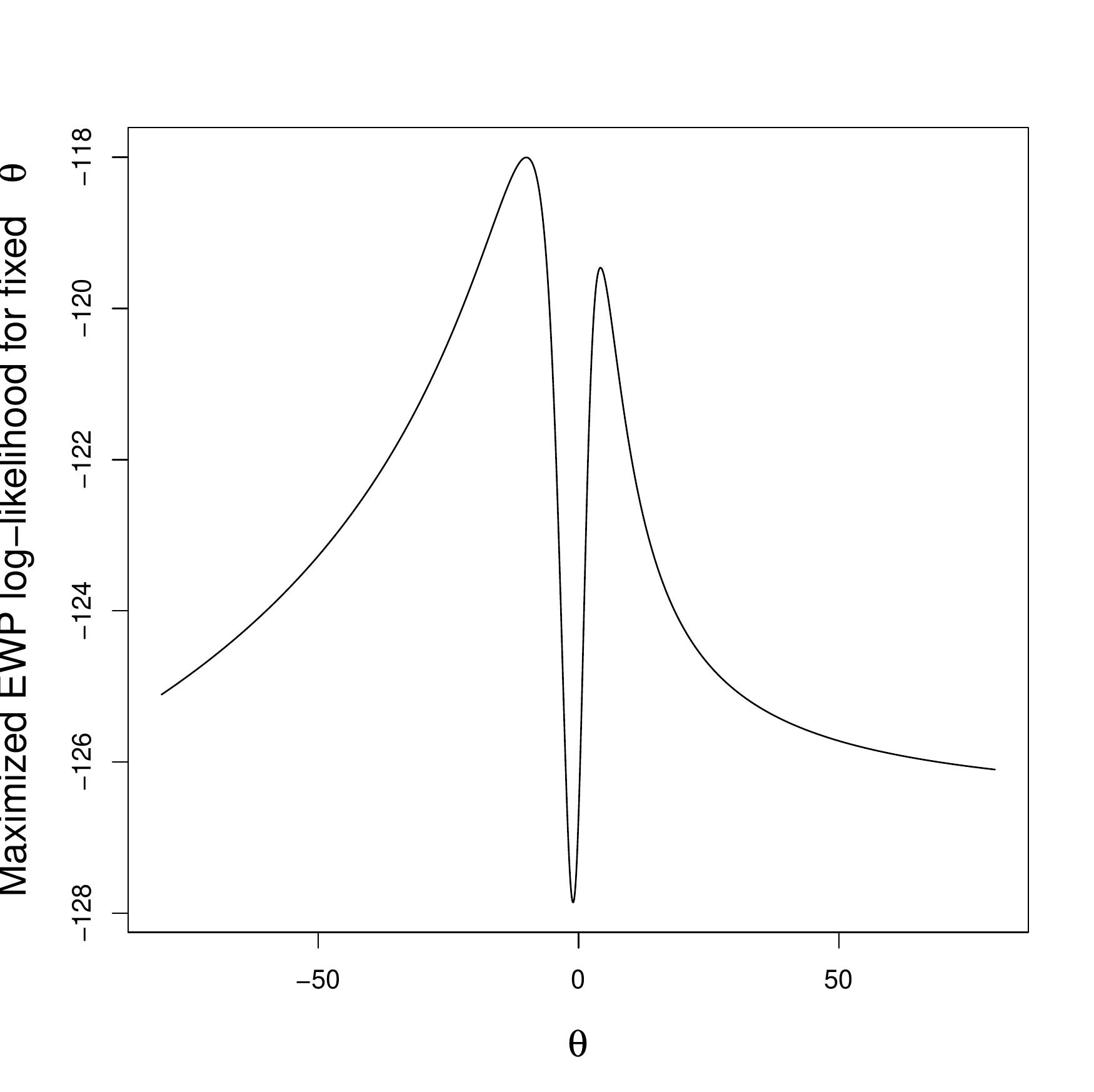}
\caption{Profile log-likelihood for the EWL, EWG, and EWP regression models.}
\label{qqplots}
\end{figure}

\begin{table}[h]
\caption{Parameter estimates (standard errors are shown in parentheses), 
maximized log-likelihood, log-likelihood ratio statistic (LR), and AIC for the Weibull, EWL, EWG, and EWP regression models.}
\begin{center}
{\small
\begin{tabular}{ccccccccc}
\hline
Model&$\widehat\beta_0$&$\widehat\beta_1$&$\widehat\beta_2$&$\widehat\alpha$&$\widehat\theta$&$\widehat\ell$&LR&AIC\\
\hline
Weibull&$-0.2750 $ &$-0.0132$&$-0.5729$&$3.2631$ &-&$-126.7286 $&-&$261.4$\\
&$(0.0989)$ & $(0.0038)$ &$(0.0555)$&$(0.1583)$  &     &&&\\
EWL&$-0.0613$ &$-0.0116$&$-0.5541$&$4.4998$&$0.9355$&$-123.9225$&$5.6122$&$257.8$\\
&$(0.1225)$ &$ (0.0038)$ &$(0.0568)$ &$(0.6588)$&$(0.0775)$&&$(0.0178)$&\\
EWG&$0.0982$ &$-0.0109 $&$-0.5876$ &$5.0520$ &$0.9455$&$-118.1125$&$17.2323$&$ 246.2$\\
&$(0.1491)$ &$ (0.0040)$ &$(0.0622)$ &$(0.3679)$&$(0.0405)$&&$(<0.001 )$&\\
EWP& $-1.1728$ &$-0.0119$&$-0.5813$  &$1.3694$&$-10.0190$ &$-118.0037$&$17.4498$&$246.0$\\
&$(0.2179)$ &$ (0.0040)$ &$(0.0615)$ &$(0.1859)$&$(3.2091)$&&$(<0.001 )$&\\
\hline
\end{tabular}
}
\end{center}
\end{table}

The parameter estimates, the value of the maximized log-likelihood $\widehat\ell$, and the Akaike Information Criterion (AIC) for the different regression models are given in Table 2. We also performed the likelihood ratio (LR) test of the null hypothesis ${\cal H}_0: \theta=0$ against ${\cal H}_1: \theta\neq 0$, i.e., the Weibull model was tested against a chosen EWPS model. Under the null hypothesis, the asymptotic distribution of the LR test statistic is $\chi_1^2$. The LR statistics used to test the Weibull model against the EWL, EWG, and EWP are given in Table 2 with their respective $p$-values in parentheses. For the usual significance levels, the LR tests rejected the Weibull distribution in favor of the EWG and EWP distributions. The AICs indicate that the EWG and the EWP models fit the data better than the EWL and Weibull models.

For diagnostic purposes, we used the quantile residual proposed by Dunn and Smyth (1996). Let $Y\sim \mbox{EWPS}(\lambda,\alpha,\theta;C)$ and let $Q_r=\Phi^{-1}(F(Y;\lambda,\alpha,\theta))$, where $\Phi(\cdot)$ is the cdf of the standard normal distribution, and $F(Y;\lambda,\alpha,\theta)$ is the cdf given in (\ref{cdf}). Then, $F(Y;\lambda,\alpha,\theta)$ is uniformly distributed in the unit interval, and $Q_r$ has a standard normal distribution. Hence, if the assumed EWPS regression model is suitable for the data, the quantile residuals defined as $Q_{r,i}=\Phi^{-1}(F(y_i;\widehat\lambda_i,\widehat\alpha,\widehat\theta))$, for $i=1,\ldots,n$, are expected to behave as iid $N(0,1)$ random variables. Here, $y_i$ denotes the $i$th observed response.

Figure 4 presents Q-Q plots of the quantile residuals for the four different fitted EWPS models. It is clear that the EWG and EWP models fit the data better than the other models and that the Weibull regression model is undoubtedly inappropriate. These findings are confirmed through the Shapiro-Wilk and Anderson-Darling normality tests, as shown in Table 3. None of the tests rejects the normality of the residuals for the EWG and EWP fitted models. The comparision of the EWP and EWG models revealed that the first appears to be the best choice: its quantile residuals agree almost perfectly with the normal quantiles, it presents the largest p-values for all of the normality tests, and it is the only model that captures the parallel system nature of the coconut fibers. 

\begin{table}[h]
\caption{Normality test statistics of quantile residuals for the fitted models (p-values are shown in parentheses).}
\begin{center}
\begin{tabular}{ccc}
\hline
Model& Shapiro-Wilk &Anderson-Darling\\
\hline
Weibull&$0.9719$&$1.1515$\\
&$(0.0002)$&$(0.0051)$\\
EWL&$0.9810$&$0.6785$\\
&$( 0.0039)$&$(0.0754)$\\
EWG&$0.9970$&$0.2565$\\
&$(0.9456)$&$(0.7201)$\\
EWP&$0.9975$&$0.2039$\\
&$(0.9778)$&$(0.8738)$\\
\hline
\end{tabular}
\end{center}
\end{table}

\begin{figure}[h]
	\centering
\includegraphics[width=.4\textwidth]{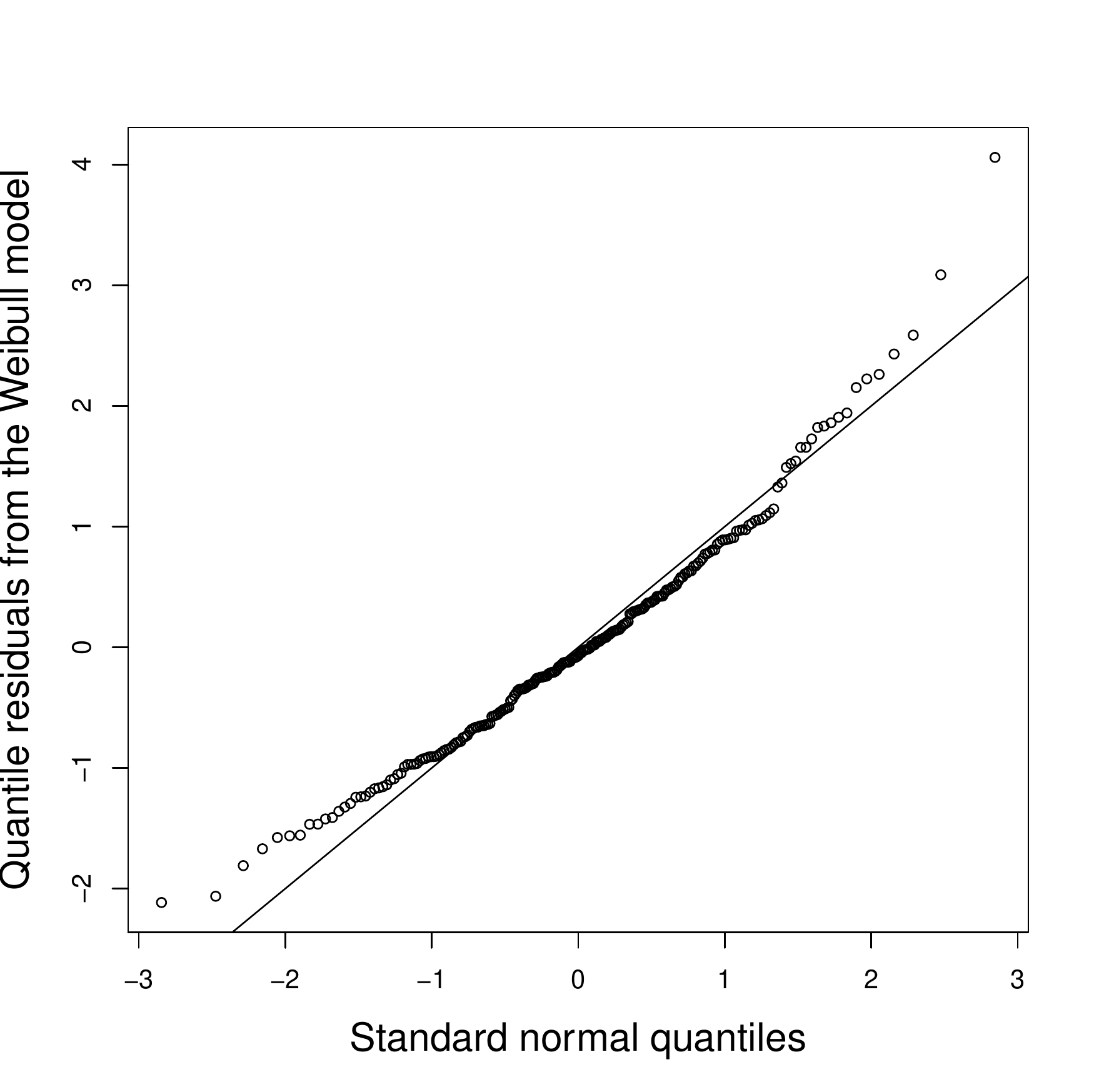}\includegraphics[width=.4\textwidth]{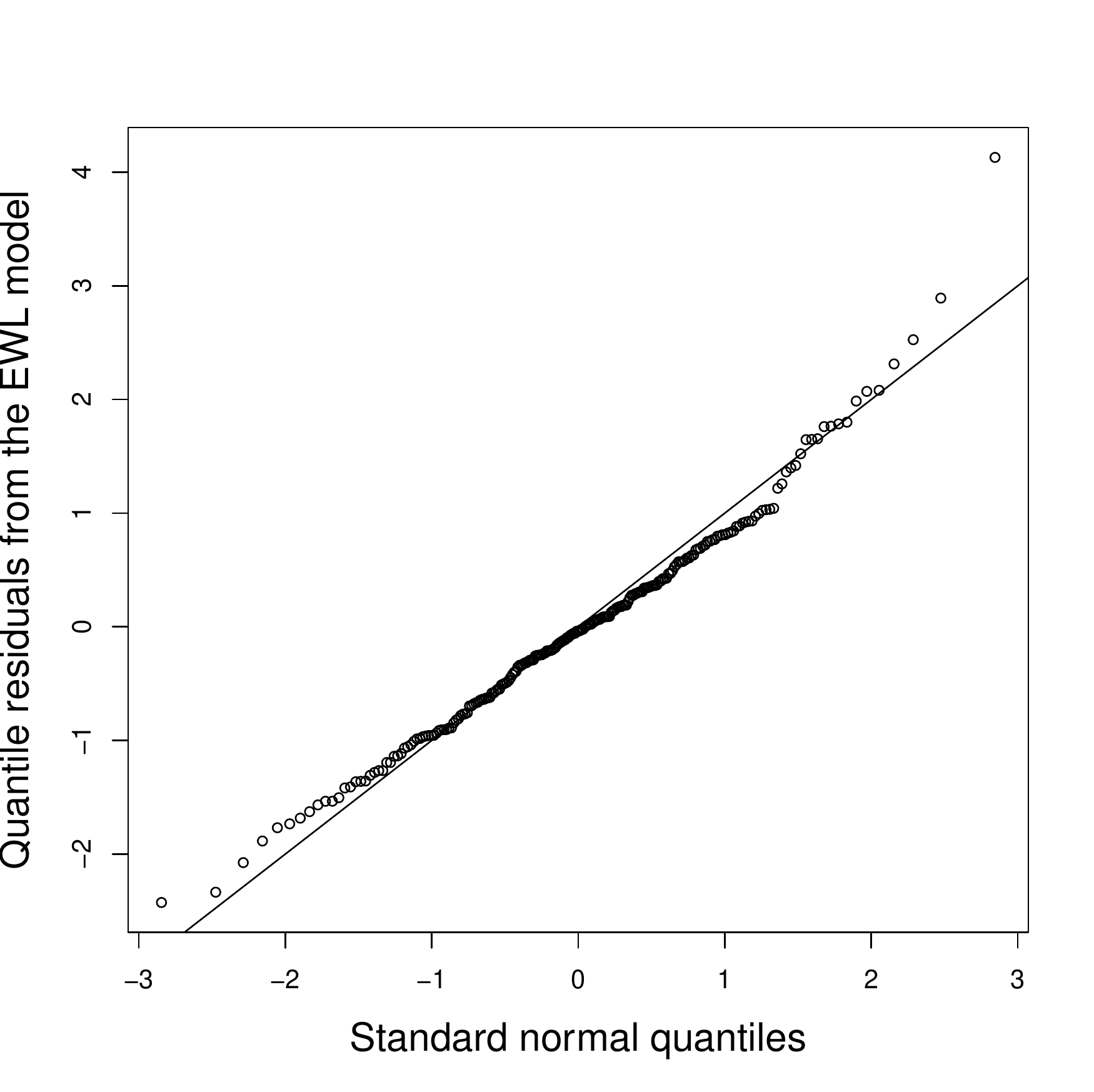}
\includegraphics[width=.4\textwidth]{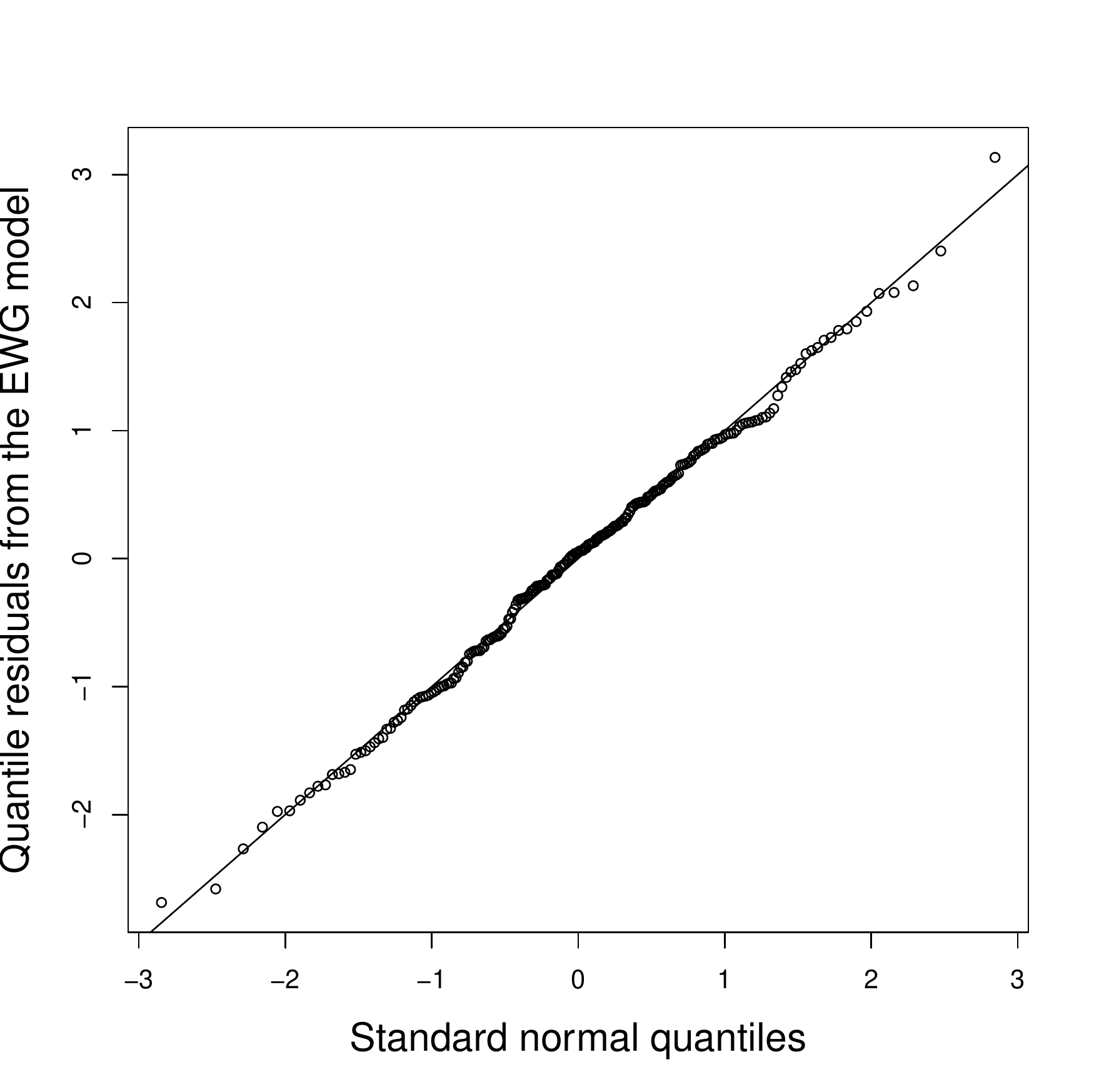}\includegraphics[width=.4\textwidth]{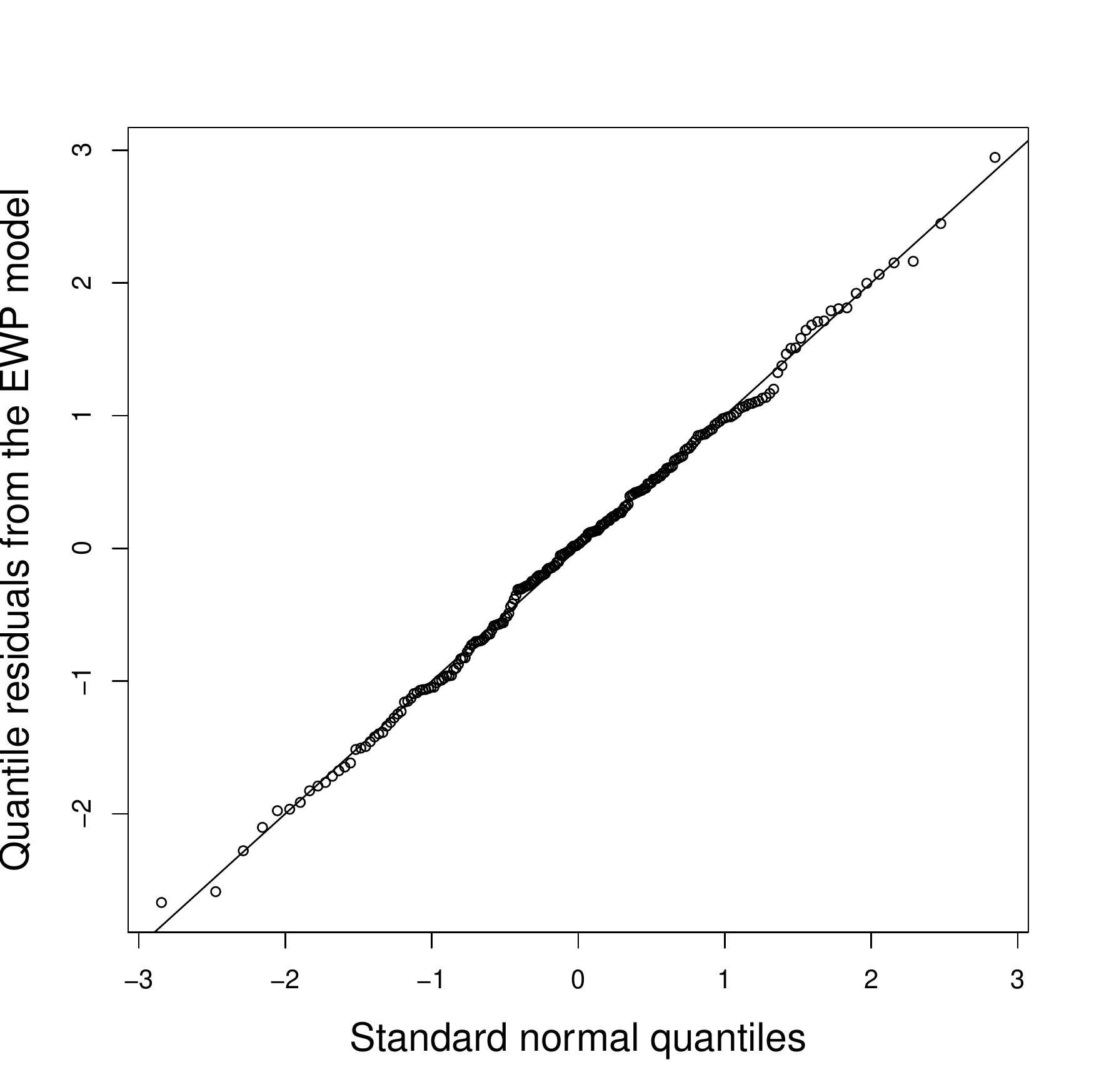}
\caption{Q-Q plot of the quantile residuals.}
\label{qqplots}
\end{figure}

Figure \ref{fit.quantis} shows plots of the fitted $0.1$, $0.5$, and $0.9$ quantiles estimated from the EWP model against the fiber diameter for fibers with a length equal to 20 mm. The solid lines are the quantiles curves, and the dashed lines bind the $95\%$ approximated confidence bands. The confidence intervals are obtained according to the asymptotic normal distribution of the estimated quantiles and the approximate variance given in (\ref{var_q}).

\begin{figure}[h]
	\centering
		\includegraphics[width=.33\textwidth]{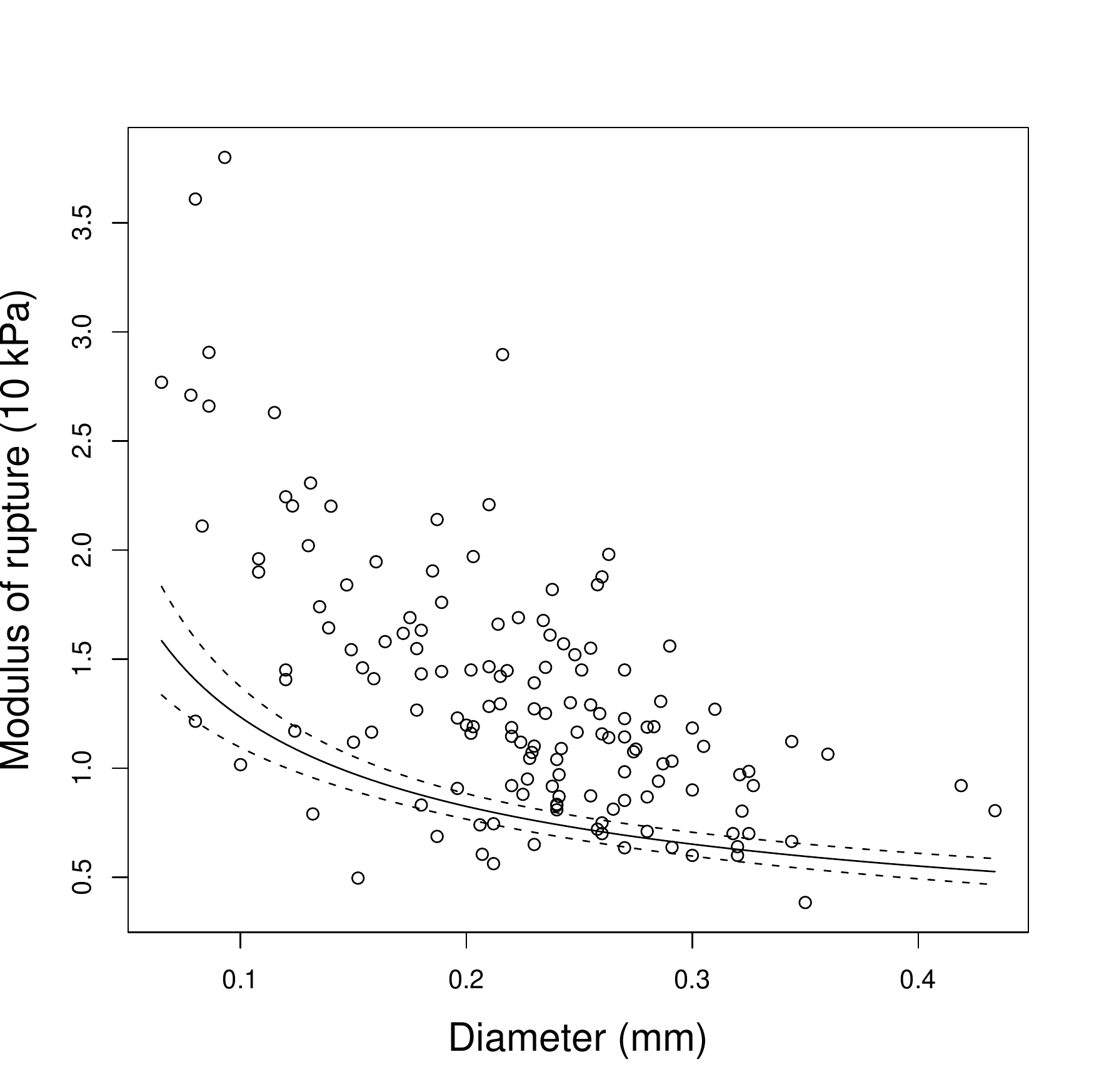}\includegraphics[width=.33\textwidth]{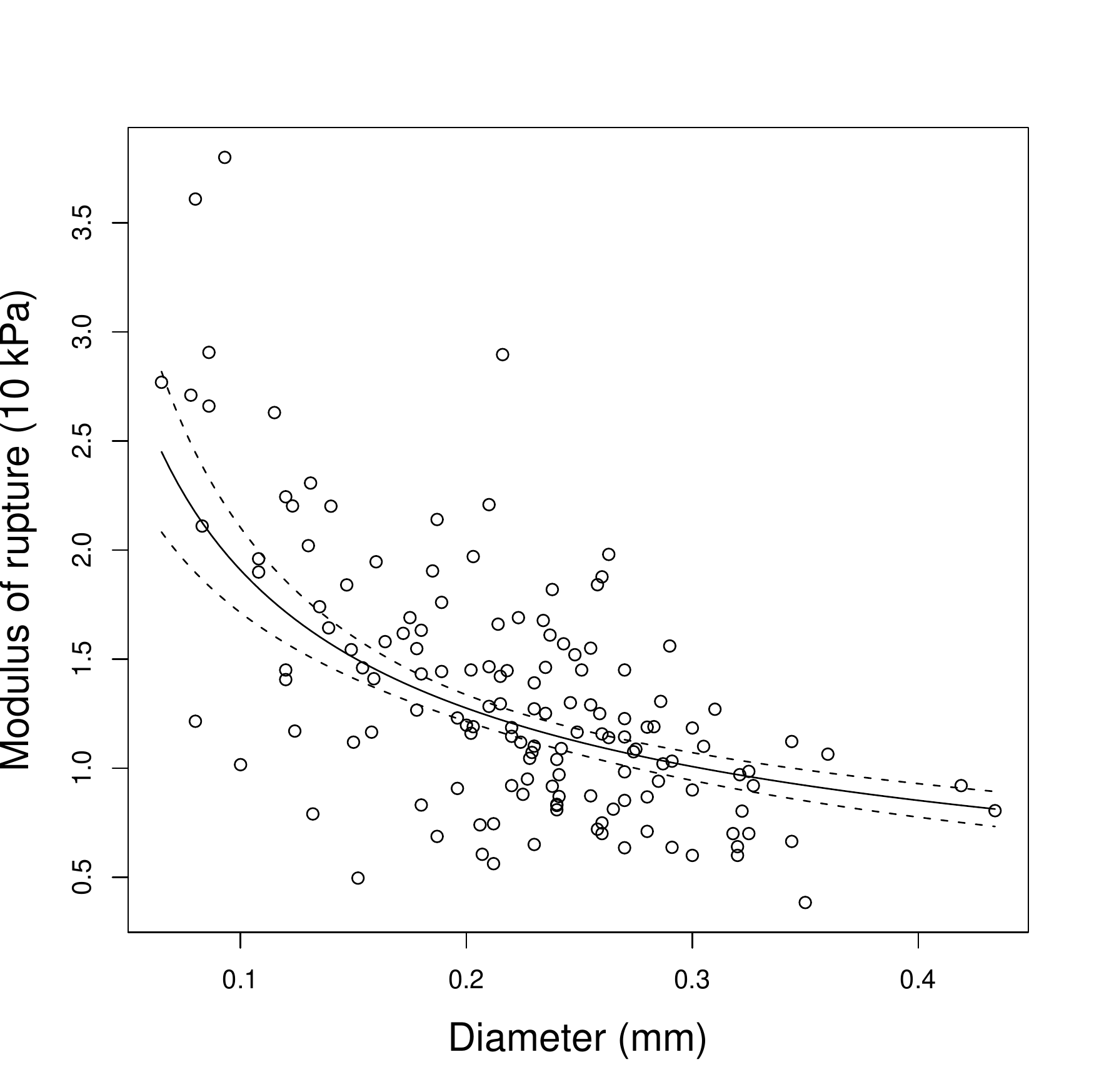}\includegraphics[width=.33\textwidth]{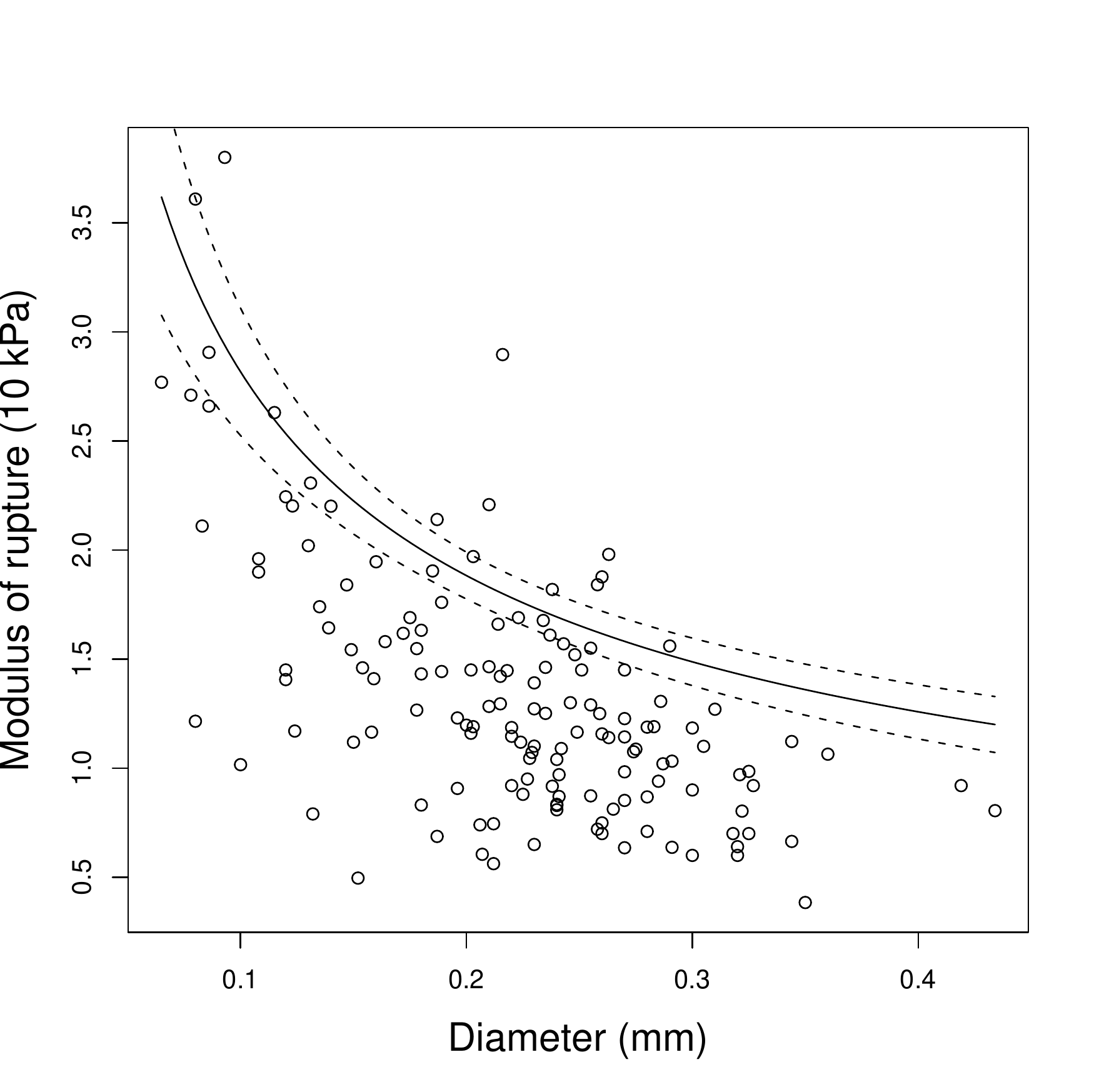}
	\label{aplicacao_graf1}\caption{Plots of the fitted EWP quantiles for $\xi=0.1$, $0.5$ and $0.9$. The solid lines are the point quantiles estimates, and the dashed lines bound the 95\% confidence region.}
\label{fit.quantis}
\end{figure}

\section{Discussion of a larger extension of the EWPS class of distributions}

The WPS distributions introduced by Morais and Barreto-Souza (2011) are based on a composition between the Weibull law with scale and shape parameters $\lambda$ and $\alpha$, respectively, and the discrete power series class of distributions with parameter $\theta$. It is well known that the parameter $\theta$ of the power series distributions is positive. In this paper, we present the extended Weibull power series (EWPS) distributions, which are an extension of the WPS distributions that accepts negative values for $\theta$. The construction of the EWPS distributions was based on the power series function $C(\cdot)$ given in (\ref{fpseries}) with radius of convergence $s$. 
For some EWPS distributions, however, the parameter space for $\theta$ can be even further extended to the interval $(s^\dag,s)$, where $s^\dag\leq -s$. For instance,  the parameter space for $\theta$ in the EWG and EWL distributions can be extended to $(-\infty,1)$. The respective EWPS distributions for the other power series reported in Table 1 do not allow a larger extension, i.e., $\theta$ cannot be smaller than $s^*$. 
It should be noted that the results presented in Section 2 may not be valid for $\theta\leq -s$.

To formalize this concept, we define $D:(s^\dag,s)\rightarrow \mathbb{R}$ as a continuous function that admits all derivatives and such that 
$D(\theta)=C(\theta)$, $\forall \theta\in (-s,s)$.
Note that $D(\theta)$ is not written as a power series when $\theta\in (s^\dag,-s)$. Taking $\mathcal{D}=\{\theta\in (s^\dag,0): D'(\theta)=0\}$, we define
\begin{equation*}\label{def_s2}
d^*=\left\{
\begin{array}[c]{cc}
	\max \mathcal{D},& {\rm if}\,\,\, \mathcal{D}\neq\emptyset\\
	s^\dag,& {\rm otherwise.}\\
\end{array}\right.
\end{equation*}
It is possible to check that the function in (\ref{dens_WPS}) is a density function for $\theta\in (d^*,0)$. After this larger extension of the parameter space, the identifiability still holds if and only if $C(\cdot)$ is not odd. The parallel system characterization discussed in Proposition 2.5 may not be valid for $\theta\leq -s$ even if it is valid for $\theta\in (-s,0)$. For the EWG and EWL distributions, the parallel system characterization for $\theta\leq -1$ is valid by taking $t(\theta)=\theta/(\theta-1)$.

Other extended classes of distributions can be constructed from the ideas presented in this paper. In fact, the Weibull distribution used for the construction of the EWPS class of distributions can be replaced by any other survival distribution, for instance, the exponential, gamma, and lognormal distributions. The results of all the propositions shown in Section 2, except for Proposition 2.4, depend on the Weibull distribution only through its survival function, which may be replaced by any other survival function. Proposition 2.4 refers to the identifiability, which will hold by replacing the Weibull survival function for any other survival function only if its respective distribution is identifiable. It is clear that the inferential methods presented in Sections 3 and 4 should be tailored for the chosen survival distribution. In this paper, the choice of the Weibull distribution was motivated by its popularity for the modeling of survival and reliability data.

\section{Concluding remarks}

We extended the Weibull power series (WPS) class of distributions proposed by Morais and Barreto-Souza (2011) such that the Weibull distribution is a special case of this new class. This extended class, which was named the extended Weibull power series (EWPS) class of distributions, is based on a composition between the Weibull and the power series distributions. The WPS and EWPS distributions are indexed by a scale parameter $\lambda$, a shape parameter $\alpha$, and a parameter $\theta$, which was inherited from the power series distributions. For the WPS distributions, $\theta$ is a positive parameter. For the EWPS distributions, the parameter space of $\theta$ was extended to include negative values.

The WPS distributions are related to series systems in which the number of components follows a power series distribution and the time to failure of each component follows a Weibull distribution. The WPS distributions exhibit a stochastic and hazard order according to the parameter inherited from the power series distribution in the construction of the WPS distributions. The hazard functions of the WPS distributions are always above the hazard function of the Weibull distribution. This is a limitation of these distributions that is eliminated when the extension of the class is considered. The hazard functions of the EWPS distributions may be below, above, or cross the hazard function of the Weibull distribution. Although the construction of this class was motivated by series and parallel systems, this model is suitable for a wide range of data with positive support. 

We proposed a regression model for the EWPS class of distributions, so-called EWPS regression models. A linear regression structure was defined for the scale parameter. We discussed estimation by the maximum likelihood approach and derived the total observed 
information matrix, which is useful for making inferences for the parameters of the regression model. We verified the asymptotic distribution for the ML estimator. In survival studies, there is an interest in the quantiles of the response variable. To meet this need, we discussed the ML estimation of quantiles. In addition, we fitted the EWP, EWG, and EWL regression models for a real data set on the tensile strength of coconut fibers of different lengths and diameters to illustrate the applicability of the EWPS regression models and presented a simple device for diagnostic purposes.

\section*{Acknowledgments}
We gratefully acknowledge the financial support from CNPq, CAPES, and FAPESP (Brazil).


\section*{Appendix}

 \renewcommand{\theequation}{A.\arabic{equation}}
 \setcounter{equation}{0}

\noindent {\bf Proof of Proposition 3.4}

\vspace{.4cm}

From the cdf $F(y;\lambda,\alpha,\theta)$ in (\ref{cdf}), it is easy to see that the EWPS distribution is not identifiable if the function $C(\cdot)$ is odd. If identifiability does not hold, there are two different parameter vectors $\Theta_1=(\alpha_1,\lambda_1,\theta_1)$ and $\Theta_2=(\alpha_2,\lambda_2,\theta_2)$ for which $f(y;\lambda_1,\alpha_1,\theta_1)=f(y;\lambda_2,\alpha_2,\theta_2)$ $\forall y>0$. Assume $\theta_2=m\theta_1$. From (\ref{expansao_dens}), we obtain
$$\sum_{i=1}^\infty\sum_{n=1}^\infty a_na_i\theta_1^{i+n}m^i g(y;\lambda_1 n^{-1/\alpha_1},\alpha_1)=\sum_{i=1}^\infty\sum_{n=1}^\infty a_na_i\theta_1^{i+n}m^n g(y;\lambda_1 n^{-1/\alpha_2},\alpha_2), \quad \forall y>0.$$

\noindent Taking the first term of the polynomial in $\theta_1$ on both the right-hand and the left-hand sides of the above equation, we have that $g(y;\lambda_1,\alpha_1)=g(y;\lambda_2,\alpha_2)$, $\forall y>0$. In addition, because the Weibull distribution is identifiable, $\lambda_1=\lambda_2$ and $\alpha_1=\alpha_2$. Then, from now on, let $\lambda=\lambda_1$ and $\alpha=\alpha_1$. From the cdf $F(y;\lambda,\alpha,\theta)$ in (\ref{cdf}) we have 
$C(\theta_2)\sum_{n}  a_n \theta_1^{n} \left(e^{-(y/\lambda)^\alpha}\right)^n=C(\theta_1)\sum_{n}  a_n \theta_2^{n} \left(e^{-(y/\lambda)^\alpha}\right)^n, \quad \forall y>0.$
From the polynomial in $e^{-(y/\lambda)^\alpha}$ we obtain 

\begin{equation}\label{div}
\left (\frac{\theta_1}{\theta_2}\right)^n=\frac{C(\theta_1)}{C(\theta_2)},\quad \forall n\in K,
\end{equation}
\noindent where $K=\{k\in\mathbb{N}: a_k>0\}$. If $K$ contains both even and odd values, then (\ref{div}) implies that $\theta_1=\theta_2$, which is an absurd due to the non-identifiability assumption. If $K$ contains only odd values, then (\ref{div}) is satisfied for $\theta_1=-\theta_2$ and the function $C(\cdot)$ is odd.


\vspace{.4cm}

\noindent {\bf Proof of Proposition 3.5}

\vspace{.4cm}


Without loss of generality, we consider $\alpha=\lambda=1$. The marginal cdf of $Z$ is given by
\begin{eqnarray}
\nonumber P(Z\leq y)&=&\frac{C(t(\theta)(1-e^{-y}))}{C(t(\theta))}\\
\nonumber &=&\frac{1}{C(t(\theta))}\sum_{n=1}^\infty \sum_{i=0}^n a_n t(\theta)^n\binom{n}{i}(-1)^i e^{-iy}\\
\nonumber &=&1-\sum_{i=1}^\infty \frac{(-1)^{i-1}}{i!C(t(\theta))}e^{-iy}\sum_{n=1}^\infty \frac{a_n n!}{(n-i)!} t(\theta)^n \\
&=&1-\sum_{i=1}^\infty \frac{(-1)^{i-1}t(\theta)^i e^{-iy}}{i! C(t(\theta))}C^{(i)}(t(\theta))\label{exp1}\\
&=& 1-\sum_{i=1}^\infty \frac{a_i (\theta e^{-y})^i}{C(\theta)}=F(y;\lambda,\alpha,\theta),\label{exp2}
\end{eqnarray}
\noindent where $F$ is the cdf of the EWPS distribution given in (\ref{cdf}). 

\vspace{.7cm}

\noindent {\bf Proof of Proposition 3.6}\\

(i) By assumption, the cdf in (\ref{exp1}) is equal to the cdf in (\ref{exp2}). Because these are both polynomials in $e^{-y}$, if $a_n=0$, then $C^{(n)}(t(\theta)=0$ for all $\theta\in (s^*,s)$, which implies that $a_m=0$ $\forall m>n$.

(ii) Equating the first coefficients of the polynomials in $e^{-y}$ in (\ref{exp1}) and (\ref{exp2}), we obtain $E(N)=a_1\theta/C(\theta)$. Then, $t(\theta)$ is the solution of $E(N)=a_1\theta/C(\theta)$. Because the expected value of a power series random variable is monotone on its parameter, this solution is unique.

(iii) We have that $\theta/C(\theta)$ is decreasing in $\theta$. Then, from the equation $E(N)=a_1\theta/C(\theta)$, $E(N)$ is increasing in $t(\theta)$ and decreasing in $\theta$. Therefore, $t(\theta)$ is decreasing in $\theta$.

\vspace{.7cm}
\newpage
\noindent {\bf Proof of Proposition 3.8}\\


Let $\theta\neq 0$. Expanding $C'(\theta e^{-(y/\lambda)^\alpha})$ and rearranging the terms of the sum, we obtain

\begin{eqnarray}\label{p_hf1}
\frac{r(y;\lambda,\alpha,\theta)}{r_0(y;\lambda,\alpha)}=\frac{\theta e^{-(y/\lambda)^\alpha} C'(\theta e^{-(y/\lambda)^\alpha})}{C(\theta e^{-(y/\lambda)^\alpha})}&=&1+\frac{\sum_{j=1}^\infty\sum_{i=j}^\infty \theta^i e^{-i(y/\lambda)^\alpha}ia_i}{C(\theta e^{-(y/\lambda)^\alpha})}.
\end{eqnarray}

The limit of (\ref{p_hf1}) when $y\rightarrow \infty$ is equivalent to the limit when $b\equiv e^{-(y/\lambda)^\alpha}$ goes to zero. Applying  L'Hôpital's rule for the second term of the right side of (\ref{p_hf1}) once, we find that (\ref{p_hf1}) goes to $1$ as $y\rightarrow\infty$ or $b\rightarrow 0$.

\vspace{.7cm}

\noindent {\bf Proof of Proposition 3.9}\\

Let $N_{\theta}\sim{\rm PS}(\theta;C)$ for $\theta>0$, and let $M_{\theta}(t)$ for $t>0$ be the moment generating function of $N$. If $0<\theta_1<\theta_2$, $M_{\theta_1}(t)\leq M_{\theta_2}(t)$. Then, from Theorem 5.1 from Shaked and Wong (1997), $Y_{\theta_2}\leq_{st} Y_{\theta_1}$ and $Y_{\theta_2}\leq_{hr} Y_{\theta_1}$. Since we have a weak convergence of $Y_\theta\rightarrow Y_0$ when $\theta\rightarrow 0^+$, where $Y_0\sim\mbox{Weibull}(\lambda,\alpha)$, the stochastic and hazard rate orders still apply to $\theta_1=0$.
For $\theta_1<0$ and under the existence of a parallel system characterization, the proof follows similarly. 

\vspace{.7cm}


\noindent {\bf Proof of Lemma 4.1}\\

\noindent (i) For $\theta\neq 0$, it is trivial to prove that all third derivatives exist. To prove that the third derivatives exist for $\theta=0$, it is sufficient to show that their limits as $\theta\rightarrow 0$ exist. We have
\begin{eqnarray*}
\frac{\partial^3\ell}{\partial \theta^3}&=&-nA(\theta)+\sum_{i=1}^n B_i(\theta) e^{-3W_i}\\
&&\rightarrow-n\frac{2a_2^3-6a_1a_2a_3}{a_1^3}+
\sum_{i=1}^n\frac{16 a_2^3-48a_1a_2a_3}{a_1^3} e^{-3W_i}, \quad \mbox{as} \quad \theta \rightarrow 0,
\end{eqnarray*}
\noindent where $$A(\theta)=\frac{1}{C(\theta)^3}\left (2C'(\theta)^3+C'''(\theta)C(\theta)^2-2C(\theta)^3\theta^{-3}-3C(\theta)C'(\theta)C''(\theta)\right )$$
\noindent and $$B_i(\theta)=\frac{1}{C'(\theta e^{-w_i})}\left (C''''(\theta e^{-w_i})C'(\theta e^{-w_i})^2+2C''(\theta e^{-w_i})^3-3C'''(\theta e^{-w_i})C''(\theta e^{-w_i})C'(\theta e^{-w_i})\right ).$$

\noindent Similarly, it can be shown that all other third derivatives exist when $\theta=0$.

Because $A(\theta)$ and $B_i(\theta)$ are continuous functions of $\theta$, there are $s^*<{s^*}'<0<s'<s$ and constants $c_1,c_2<\infty$ such that $|A(\theta)|<c_1$ and $|B_i(\theta)|<c_2$, $\forall i$, $\forall \theta\in [{s^*}',s']$. By the triangle inequality, 

$$\left |\frac{\partial^3\ell}{\partial \theta^3}\right | < n c_1+c_2\sum_{i=1}^n e^{-3W_i}<n(c_1+c_2).$$

\noindent Then, $\partial^3\ell/\partial\theta^3$ is dominated by the integrable function $n(c_1+c_2)$, which does not depend on the parameters. The proof for the other third derivatives is similar.

\vspace{.5cm}

\noindent (ii) The proof follows using the expansions $C'(\theta)=\sum_{n=1}^\infty n a_n \theta^{n-1}$ and $C''(\theta)=\sum_{n=1}^\infty n(n-1) a_n \theta^{n-2}$ and permuting the signs of integral and sum.

\vspace{.5cm}

\noindent (iii) The proof follows from (i) and (ii) and from the Dominated Convergence Theorem.

\vspace{.5cm}

\noindent (iv) The components of the score vector and the components of the total observed information matrix depend on the response variable only through $W_i$ for $i=1,\ldots,n$. Under $\Theta=\Theta^{(0)}$, the random variables $W_1,\ldots,W_n$ are iid. Then, the results follow from (ii), (iii), and the law of large numbers.


\begin{thebibliography}{30}

\bibitem{bara2004} Barakat, H. M., El-Shandidy, M. A., 2004. Order statistics with random sample size. \emph{International Journal of Statistics}, 62, 233-246.
\bibitem{barn1964} Barndorff-Nielsen, O., 1964. On the limit distribution of the maximum of a random number of independent random variables. \emph{Acta Mathematica Academiae Scientiarum Hungarica}, 15, 399-403.
\bibitem{cooner2007} Cooner, F., Banerjee, S., Carlin, B.P., Sinha, D., 2007. Flexible cure rate modeling under latent activation schemes. \emph{Journal of the American Statistical Association}, 102, 560-572.
\bibitem{crescenzo2011} Crescenzo, A., Pellerey, F., 2011. Stochastic comparisons of series and parallel systems with randomized independent components. \emph{Operations Research Letters}, 39, 380-384.
\bibitem{dunn1996} Dunn, P. K., Smyth, G., 1996. Randomized Quantile Residuals. \emph{Journal of the Computational and Graphical Statistics}, 5, 236-244.
\bibitem{kus2007} Kus, C., 2007. A new lifetime distribution. \emph{Computational Statistics \& Data Analysis}, 51, 4497-4509.
\bibitem{lehmman1998} Lehmann, E. L., Casella, G., 1998. \emph{Theory of Point Estimation}. Springer, New York. Second edition.
\bibitem{marshall1997} Marshall, A. W., Olkin, I., 1997. A new method for adding a parameter to a family of distributions with application to the exponential and Weibull families. \emph{Biometrika}, 84, 641-652.
\bibitem{Alice} Morais, A. L., Barreto-Souza W., 2011. A compound class of Weibull and power series distributions. \emph{Computational Statistics and Data Analysis}, 55, 1410-1425.
\bibitem{naga2012} Nakagawa, T., Zhao, X., 2012. Optimization Problems of a Parallel System With a Random Number of Units. \emph{IEEE Transactions on Reliability}, 61, 543-548.
\bibitem{noack1950}Noack, A., 1950. A class of random variables with discrete distributions. \emph{Annals of Mathematical Statistics}, 21, 127-132.
\bibitem{pksenbook2009} Sen, P., K., Singer, J. M, Lima, A. C. P., 2009. \emph{From Finite Sample to Asymptotic Methods in Statistics}. Cambridge University Press, New York. First edition.
\bibitem{shaked2007} Shaked, M., Wong, T., 1997. Stochastic Orders Based on Ratios of Laplace Transforms. \emph{Journal of Applied Probability}, 34, 420-425.
\bibitem{tomczak2007} Tomczak, F., Satyanarayana, K. G., Sydenstricker, T. H. D., 2007. Studies on lignocellulosic fibers of Brazil: part II - morphology and properties of Brazilian coconut fibers. \emph{Composites Part A: Applied Science and Manufacturing}, 38, 1710-1721.
\bibitem{Yak1993} Yakovlev, A. Y., Asselain, B., Bardou, V. J., Fourquet, A., Hoang, T., Rochefediere, A. and Tsodikov, A. D., 1993. A simple stochastic model of tumor recurrence and its applications to data on premenopausal breast cancer. \emph{Biometrie et Analyse de Dormees Spatio-Temporelles}, 12, 66-82.



\end{thebibliography}
\end{document}